\documentclass[a4paper,10pt]{amsart}
\usepackage[utf8]{inputenc}
\usepackage[T1]{fontenc}
\usepackage{graphicx} 
\usepackage{mathtools}
\usepackage{color}
\usepackage{enumitem}
\usepackage{cite}
\usepackage{url}
\usepackage{amsmath}
\usepackage{amssymb}
\usepackage{amsthm}
\usepackage{mathrsfs}
\usepackage{mathtools}

\newcommand{\ena}[1]{\mathrm{e}^{#1}}

\DeclareMathOperator{\supp}{\mathrm{supp}}
\newcommand{\dd}{\mathrm{d}}

\newcommand{\N}{\mathbb{N}}
\newcommand{\C}{\mathbb{C}}
\newcommand{\R}{\mathbb{R}}
\newcommand{\Z}{\mathbb{Z}}
\newcommand{\iu}{\mathrm{i}} %imaginary unit
\newcommand{\HH}{\mathscr{H}} % continuum Hilbert space 
\newcommand{\Hh}{\mathscr{H}_h} %discrete Hilbert space 
\newcommand{\hHh}{\hat{\mathscr{H}}_h} %and its F. image
\newcommand{\F}{\mathcal{F}}
\newcommand{\Fh}{\textup{F}_h}
\newcommand{\B}{\mathscr{B}}

\newcommand{\GG}{\mathscr{G}_\infty} % continuum Hilbert space a la Barker
\newcommand{\Gn}{\mathscr{G}_n} %discrete Hilbert space a a la Barker
\newcommand{\res}{\mathrm{res}} %inductive resolution
\newcommand{\Sub}{\mathscr{S}} % dense subspace of the continuum Hilbert space a la Barker
\newcommand{\NN}{\mathcal{N}} %directed set

\newcommand{\hn}{h_{n}}
\newcommand{\mass}{\text{\sffamily\itshape m}}

\newcommand{\wtlim}{\mathop{\widetilde{\lim}}}
\DeclareMathOperator*{\wlim}{w-lim}
\DeclareMathOperator*{\slim}{s-lim}
\DeclareMathOperator{\Dom}{Dom}

\renewcommand{\phi}{\varphi} 
\renewcommand{\epsilon}{\varepsilon}

\newtheorem{theorem}{Theorem}[section]
\newtheorem{lemma}[theorem]{Lemma}
\newtheorem{corollary}[theorem]{Corollary}
\newtheorem{proposition}[theorem]{Proposition}
\newtheorem{definition}[theorem]{Definition}
\theoremstyle{definition}
\newtheorem{remark}[theorem]{Remark}

\theoremstyle{definition}
\newtheorem{assumptions}{Assumptions}
\newtheorem{example}{Example}

\title{Continuum limit of discretized matrix-valued Fourier multipliers}

\date{July 23, 2026}
\begin{document}

\author[R. Karapetyan]{Ruben Karapetyan}
\address{Department of Computer Science, Faculty of Electrical Engineering\\
Czech Technical University in Prague \\
Karlovo náměstí 13, 120 00, Prague, Czechia \\
E-mail: {\tt karaprub@fjfi.cvut.cz}}

\author[M. Tu\v{s}ek]{Mat\v{e}j Tu\v{s}ek}
\address{Department of Mathematics, Faculty of Nuclear Sciences and Physical Engineering\\
Czech Technical University in Prague \\
Trojanova 13, 120 00, Prague, Czechia \\
E-mail: {\tt matej.tusek@fjfi.cvut.cz}}

\begin{abstract}
Building upon a recent result by H. Cornean, H. Garde, and A. Jensen concerning continuum limits of discrete Dirac operators, we extend the analysis to a wide class of block operator matrices. This class includes, among others, the bilayer graphene Hamiltonian. Our main goal is to find norm estimates for the difference between the resolvents of continuous operators and their discrete counterparts embedded in the continuum in a specific way. While some discretization schemes lead directly to convergence in the generalized norm resolvent sense as the mesh parameter tends to zero, others require the addition of a suitable correction term to ensure the convergence.
\end{abstract}

\maketitle

\section{Introduction}
Let $H_0$ be a linear operator on the Hilbert space $L^2(\R^m;\C^l)$ and $H_h$ its discrete analogue in the Hilbert space $\ell^2(h\Z^m;\C^l)$, where $h>0$ can be understood as a mesh-size parameter. Naively, one would expect that $H_h$ is somehow close to $H_0$ as $h$ tends to zero. Since $H_h$ and $H_0$ act on different Hilbert spaces, it is convenient to introduce the identification operators $J_h:\, \ell^2(h\Z^m;\C^l)\to L^2(\R^m;\C^l)$ and $K_h$ acting in the opposite direction  to inspect their closeness. Still, the domains of $J_h H_h K_h$ and $H_0$ may be different, so it seems reasonable to compare the resolvent difference 
\begin{equation} \label{eq:res_diff}
 (J_h H_h K_h-z)^{-1}-(H_0-z)^{-1}
\end{equation} 
at a common point $z$ of the resolvent sets. If \eqref{eq:res_diff} converges to zero in sufficiently strong operator topology, one can conclude that the spectrum of $H_h$ or a part of it converges to that of $H_0$.

To our best knowledge, Nakamura and Tadano obtained the first results for the uniform convergence of \eqref{eq:res_diff}. Namely, they dealt with the case when $H_0=-\Delta+V$, where $V$ belongs to a large class of potentials \cite{NaTa_21}. Their results were then applied to quantum graph Hamiltonians that can be related to discrete Schr\"{o}dinger operators on a lattice \cite{ExNaTa_22}. A large class of non-negative Fourier multipliers was considered by Cornean, Garde, and Jensen in \cite{CoGaJe_21}. The same authors then studied the Dirac operators in various dimensions as prominent examples of operators that are unbounded from both sides \cite{CoGaJe_23}. They observed that, except for the one-dimensional setting and a particular choice of the discretization (backward/forward differences), one always has to add a correction term to $H_h$ in order to get the uniform convergence of \eqref{eq:res_diff}. Without the correction term, \eqref{eq:res_diff}  converges to zero only strongly. 
The latter was observed also by Schmidt and Umeda \cite{SchUm_23,SchUm_25}, who used a different  identification operator $J_h$ than the above mentioned authors. The issue with discretizing the Dirac operator is known in lattice quantum field theory as the fermion doubling problem \cite[Chap.~4]{Ro_12}.
See \cite{Na_24} for a mathematical treatise on how to overcome it. The case when $H_0$ is an elliptic operator was studied in \cite{MiNaTa_24}. The Dirac--Hodge operator and its discrete version were considered in \cite{MiPa_23}. Remarkably, no correction term was needed in the latter setting; one can refer to Subsection \ref{sec:correction} for a possible reasoning. Furthermore, employing odd or even reflections, one can translate the convergence  results for the Laplacian on the whole Euclidean space to the corresponding results for the Dirichlet or Neumann Laplacian on a half-space \cite{CoGaJe_23II}. 

Since the relation between a continuous model and its discrete version is a fundamental question in mathematical modelling and numerical mathematics, and consequently there exists an endless landscape of literature, we have restricted the overview above to the very setting of the current paper and we will add just a few closely related results to give readers a taste of the directions they can take. Strong resolvent convergence of \eqref{eq:res_diff} to zero for Schr\"{o}dinger operators on different lattices was considered in \cite{IsJe_22}. A class of non-linear Schr\"{o}dinger equations was approximated by their discrete analogues using a finite element method in \cite{BaPe_10}. Finally, lattice models and their \emph{exact} continuous counterparts were reviewed and studied in \cite{FuKr_26}.

The main aim of the present paper is to show the uniform convergence of \eqref{eq:res_diff} to zero for a large class of operator matrices $H_0$ of the form
\begin{equation*}
H_0=\F^*G_0\F \quad \text{with} \quad G_0=G_0(\xi):=\begin{pmatrix}
 G_{11}(\xi) & G_{12}(\xi)\\
 G_{12}(\xi)^* & -G_{11}(\xi)
 \end{pmatrix},
\end{equation*}
where $\F$ denotes the Fourier transform and $G_{11}$ and $G_{12}$ are $d\times d$ matrix valued functions on $\R^m$.
This class includes the Dirac operators, the bilayer graphene Hamiltonian, many other operators with super-symmetry, and effectively also all the scalar operators considered in \cite{CoGaJe_21} (see Subsection \ref{sec:examples}).  We will work with two variants of discrete operators; both of them of the form
\begin{equation*}
H_h=\Fh^*\textup{G}_h\Fh \quad \text{with} \quad \textup{G}_h=\textup{G}_h(\xi):=G_0(f_h(\xi_1),\ldots, f_h(\xi_m))
\end{equation*}
for $\xi\in(-\pi/h,\pi/h)^m$. Here, $\Fh$ stands for the discrete Fourier transform and $f$ is such that $\Fh^*f_h\Fh=-\iu \textup{d}_h$ (in fact, for the second variant, only $(\Fh^*f_h\Fh)^2=(-\iu\textup{d}_h)^2$ holds), where $\textup{d}_h$ is a sort of a discrete version of the first derivative. This choice is motivated by the fact that $\F^*\xi_j \F=-\iu\partial_{x_j}$.

For $f_h(\xi_j)=\sin(h\xi_j)/h$, we will show that \eqref{eq:res_diff} converges to zero only strongly and that  we have to add a correction term to $H_h$ to ensure the uniform convergence. Note that, in this case, $\iu \Fh^*f_h\Fh$  is just the symmetric difference, which is second order accurate but may suffer from the grid decoupling (checker-board effect). This unwanted effect can be cured by working with staggered grids, which has been done for the Dirac operators in \cite[Section 5]{Na_24}. 

The other variant of $f_h$ is derived from the square-root of the standard second order difference, i.e.
 $$f_h(\xi_j)=\frac{2}{h}\Big|\sin\Big(\frac{h}{2}\xi_j\Big)\Big|.$$
Since now $\lim_{h\to 0+}f_h(\xi_j)=|\xi_j|$, this choice is reasonable only when
\begin{equation} \label{eq:G_0_symm}
G_0(\xi)=G_0(|\xi_1|,\ldots,|\xi_m|).
\end{equation}
If this is the case, we will prove that \eqref{eq:res_diff} converges uniformly to zero. Let us emphasize that $\Fh^*f_h\Fh$ is now a non-local operator. 

The purpose of the present paper is two-fold. First, as described above, we extend the existing results in several ways. For example, in comparison to \cite{CoGaJe_21} we can deal with scalar multipliers  that do not obey \eqref{eq:G_0_symm}. However, then we have to use  a different discretization scheme (derived from the symmetric differences) which exhibits poorer convergence properties. Still, we solve this issue by finding proper corrections. Next, concerning operator valued matrices, only the case of the Dirac operators has been studied so far.  Our class is much wider. Moreover, although the symbols of the Dirac operators do not obey \eqref{eq:G_0_symm}, many other symbols do. We can then apply the second discretization scheme (derived from the square root of the standard second differences), thereby obtaining a uniform convergence result without the need for any correction. Second, since we can adopt many auxiliary statements almost literally from the previous works, we can focus on the fundamental ideas in the proofs and dedicate more space to generalizations and making new connections. Therefore, our paper is written, in part, in an expository style. In particular, we noticed that the uniform convergence of \eqref{eq:res_diff} is an instance of the so-called \emph{generalized norm resolvent convergence} and  \emph{QUE-generalized norm resolvent convergence}. The first concept was developed by Weidmann and later generalized by Post, who also introduced the latter concept; see Section \ref{sec:convergence} for an overview and related references. In our setting, both concepts are equivalent, including the convergence speed, and yield convergence results for spectra and operator functions.
We further relate these notions to Barker's concept of convergence  \cite{Ba_01, Ba2_01, Ba3_01} and show that the choice of identification operators from \cite{CoGaJe_21, CoGaJe_23} fits into Barker's framework.
Finally, for these identification operators, we will prove that, in the self-adjoint setting, the generalized \emph{strong} resolvent convergence is equivalent to convergence of the resolvents in the sense of Barker.

The paper is organized as follows. Section \ref{sec:preliminaries} is of preliminary nature, and beside fixing the notation and conventions, its main purpose is to introduce the operators $J_h$ and $K_h$. In Section \ref{sec:main}, we will formulate and prove our main results, i.e. Theorems \ref{theo:main} and \ref{theo:strong_conv}. Various concepts of convergence of linear operators in varying Hilbert spaces are reviewed and compared in Section \ref{sec:convergence}.

\section{Preliminaries} \label{sec:preliminaries}

\subsection{Notation}
We will denote the norm of a Banach space $\mathscr X$ by $\|.\|_{\mathscr{X}}$, but when there is no risk of confusion we will write just $\|.\|$. If $\mathscr{X}$ is $\C^l$ or $\R^l$ with the euclidean norm, we will always write $|.|$ instead $\|.\|_{\mathscr{X}}$. The associated scalar product will be denoted by a dot. If $k\equiv(k_1,\ldots,k_m)\in\Z^m$ is a multi-index  then $|k|:=\max\{|k_1|,\ldots,|k_m|\}$.  For the identity operator on $\mathscr{X}$, we will use the symbol $I_{\mathscr{X}}$. If $A$ is a linear operator on $\mathscr{X}$ and $\lambda$ belongs to its resolvent set, we will abbreviate  $(A-\lambda I_\mathscr{X})^{-1}$ to $(A-\lambda)^{-1}$. Letters $C$ and $K$ will be reserved for positive constants whose values may differ from inequality to inequality or from line to line.

\subsection{Fourier transforms}
Given $m,l\in\N$, and $h>0$, we will write
\begin{eqnarray*}
&\HH^l:=L^2(\R^m,\dd x;\C^l),\, \hat\HH^l:=L^2(\R^m,\dd\xi;\C^l),\\ &\Hh^l:=\ell^2(h\Z^m;\C^l),\, \hHh^l:=L^2((-\frac{\pi}{h},\frac{\pi}{h})^m,\dd\xi;\C^l).
\end{eqnarray*}
with the convention that for $l=1$ we will omit the upper index. We will always use the following norm on $\Hh^l$,
\begin{equation*}
\|u\|_{\Hh^l}:=\Big(h^m\sum_{k\in\Z^m}|u_k|^2\Big)^{1/2}.
\end{equation*}
The Fourier--Plancherel transform $\F:\HH\to\hat\HH$ is defined as usual, i.e.,
\begin{equation*}
        (\F f)(\xi) \equiv \hat{f}(\xi)=
        \frac{1}{(2\pi)^{m/2}}
        \int_{\R^m} \ena{-\iu x \cdot \xi } f(x) \dd x.
\end{equation*}
For vector valued functions, we extend the above definition component-wise and denote the transform by the same letter $\F:\HH^l\to\hat\HH^l$. The (unitary) discrete Fourier transform $\Fh:\Hh\to\hHh$ is assumed to acts as follows
\begin{equation*}
(\Fh u) (\xi)
        = 
        \frac{h^m}{(2\pi)^{m/2}}
        \sum_{k \in \Z^m}
        u_k \ena{-\iu h k \cdot \xi}.
\end{equation*}
Its adjoint is given by
\begin{equation*}
(\Fh^* g)_k = 
        \frac{1}{(2\pi)^{m/2}}
        \int_{(-\frac{\pi}{h}, \frac{\pi}{h} )^m}
        \ena{\iu hk\cdot \xi} g(\xi) \dd\xi.
\end{equation*}
Again, we extend naturally these transforms to $\Fh:\Hh^l\to\hHh^l$ and $\Fh^*:\hHh^l\to\Hh^l$ by taking their action component-wise.

\subsection{Difference operators}
The forward, backward, and symmetric difference operators in the $j$-th direction will be denoted by $d_{h,j}^+,\, d_{h,j}^-,$ and $d_{h,j}$, respectively. Recall that they are bounded operators in $\Hh$ defined as follows 
\begin{gather*}
(d_{h,j}^+u)_k:=\frac{u(k+e_j)-u(k)}{h},\quad (d_{h,j}^-u)_k:=\frac{u(k)-u(k-e_j)}{h}\\
(d_{h,j}u)_k:=\Big(\frac{1}{2}(d_{h,j}^+ + d_{h,j}^-)u\Big)_k=\frac{u(k+e_j)-u(k-e_j)}{2h},
\end{gather*}
where $e_j$ stands for the $j$-th vector of the standard basis of $\R^m$. Their Fourier symbols are
\begin{equation*}
\frac{\ena{\iu h\xi_j}-1}{h},\quad \frac{1-\ena{-\iu h\xi_j}}{h},\quad \text{and} \quad  \frac{\iu\sin(h\xi_j)}{h},
\end{equation*}
respectively. In the one-dimensional setting, we will omit the subindex $j$ in the definitions of difference operators. The standard second difference operator is defined as 
$d_{h,j}^+d_{h,j}^-=d_{h,j}^-d_{h,j}^+$ (which is different from $d_{h,j}^2$)  and the standard discretized Laplacian as
\begin{equation} \label{eq:disc_laplacian}
\Delta_h=\sum_{j=1}^m d_{h,j}^+d_{h,j}^-.
\end{equation}
The Fourier symbol of the standard second difference operator is given by 
\begin{equation} \label{eq:std_2nd_diff}
-\Big(\frac{2}{h}\sin\Big(\frac{h}{2}\xi_j\Big)\Big)^2.
\end{equation}

\subsection{Embedding and discretization operators} \label{subsec:JK_operators}
To compare operators on the "discrete space" $\Hh^l$ to operators on the "continuum space" $\HH^l$ we will use the embedding and discretization operators defined in the same manner as in \cite{CoGaJe_21} and denoted by $J_h$ and $K_h$, respectively. Therein, $l=1$;  for vector-valued sequences and functions, we will let $J_h$ and $K_h$ act component-wise and, with a slight abuse of notation, denote the resulting operators by the same letters. 

First, note that there exist \emph{generating functions} $\phi_0,\,\psi_0\in\HH$ such that for every $h>0$ the sequences $(h^{-m/2}\phi_{h,k})_{k\in\Z^m}$ and $(h^{-m/2}\psi_{h,k})_{k\in\Z^m}$, where
\begin{equation}  \label{eq:generating_functions}
    \phi_{h,k} (x) =
    \phi_0 ( (x - h k)/h), \quad
     \psi_{h,k} (x)  = \psi_0 ( (x - h k)/h)\quad (\forall k\in\Z^m),
\end{equation}
are biorthogonal and have the Riesz property with $h$-independent Riesz constants, i.e.,
\begin{eqnarray*}
&h^{-m}\langle\phi_{h,k},\psi_{h,n}\rangle_{\HH}=\delta_{k,n},\\
&A_\phi\sum_{k\in\Z^m}|c_k|^2\leq \|\sum_{k\in\Z^m}c_kh^{-m/2}\phi_{h,k}\|^2\leq B_\phi\sum_{k\in\Z^m}|c_k|^2,\\
&A_\psi\sum_{k\in\Z^m}|c_k|^2\leq \|\sum_{k\in\Z^m}c_kh^{-m/2}\psi_{h,k}\|^2\leq B_\psi\sum_{k\in\Z^m}|c_k|^2,
\end{eqnarray*} 
for all $(c_k)_{k\in\Z^m}\in \ell^2(\Z^m)$ with $h$-independent $A_\phi,B_\phi,A_\psi,B_\psi>0$. Moreover, it is possible  to choose $\phi_0, \psi_0$ essentially bounded and such that 
\begin{equation} \label{eq:supports_ass}
\begin{aligned}
    &
    \text{supp} (  \hat{\phi}_0 ) \subset [ - 3 \pi /2 , 3 \pi /2]^m  \text{ and } | \hat{\phi}_0(\xi) | \geq C_0 \ \ \forall \xi \in [ -  \pi /2 , \pi /2]^m,
    \\
    &
    \text{supp} ( \hat{\psi}_0 ) \subset [ - 3 \pi /2 , 3 \pi /2]^m  \text{ and } | \hat{\psi}_0(\xi) | \geq C_0 \  \ \forall \xi \in [ -  \pi /2 , \pi /2]^m,
\end{aligned}
\end{equation}
cf. \cite[Section 2]{CoGaJe_21}.

Next, given any pair $\phi_0,\psi_0$ of functions with the properties above, we define
the \emph{embedding operator} $J_h : \Hh^l \to \HH^l$ as
\begin{equation*} 
    (J_h u) (x) := \sum_{k \in \Z^m} \phi_{h, k} (x) u_k \quad (\forall x\in\R^m)
\end{equation*}
and the \emph{discretization operator} $K_h : \HH^l \to \Hh^l$  as
\begin{equation*}
    (K_h f)_{k} := h^{-m} \langle \psi_{h,k} , f  \rangle_{\HH} \quad (\forall k \in \Z^m),
\end{equation*} 
where $\langle \psi_{h,k} , f  \rangle_{\HH}\in\C^l$ should be understood as the inner products of $\psi_{h,k}$ with the components of $f$. The Riesz property implies that
\begin{equation*} 
    % \label{eq:norm_embedding}
    \sup_{h>0} \|J_h\| \leq \sqrt{B_{\phi}} <+\infty.
\end{equation*}
Moreover, since, $K_h=(\tilde J_h)^*$, where 
\begin{equation*}  
    (\tilde J_h u) (x) := \sum_{k \in \Z^m} \psi_{h, k} (x) u_k \quad (\forall x\in\R^m),
\end{equation*}
we also have 
\begin{equation*}
    \sup_{h>0} \|K_h\| \leq \sqrt{B_{\psi}} < +\infty.
\end{equation*}
Due to biorthogonality, $K_h J_h = I_{\Hh^l}$. On the other hand, 
 \begin{align*}
     J_h K_h =  
     \sum_{k \in \Z^m} \frac{1}{h^m} \langle \psi_{h,k} , .\, \rangle \varphi_{h,k}
 \end{align*}
is a (not necessarily orthogonal) projection into $\HH^l$, which is never the identity, because $K_h$ is not injective.

Finally, in Section \ref{sec:convergence}, the special case when $\phi_0=\psi_0$ will be of interest. Then the sequence $(h^{-m/2}\phi_{h,k})_{k\in\Z^m}$ is orthonormal in $\HH$ (but it never forms a total set). Consequently, $A_\phi=B_\phi=1$ and $J_h$ is an isometry, in particular, $\|J_h\|=1$. Moreover, $\tilde J_h=J_h$, and therefore, $K_h=(J_h)^*$, which yields $\|K_h\|=1$. Finally, $J_hK_h$ is an orthogonal projection. 
Below, we provide a simple example of such a generating function.

\begin{example}
Put $\phi_0(x_1,\ldots,x_m)=\psi_0(x_1,\ldots,x_m)=\Pi_{j=1}^m f(x_j)$, where
$$f(x_j):=\frac{\sin(\pi x_j)}{\pi x_j}$$
is Whittaker's cardinal function, see \cite[Section 1.10]{Bowles}.
It is straightforward to check that $\phi_0$ is bounded and $(h^{-m/2}\phi_{h,k})_{k\in\Z^m}$ is an orthonormal set in $\HH$. Moreover, one has $\hat{\phi}_0(\xi)=(2\pi)^{-m/2}\chi_{(-\pi,\pi)^m}(\xi)$, where $\chi_A$ is the indicator function of a set $A$. Therefore, conditions \eqref{eq:supports_ass} are satisfied. Note that $\supp(\hat{\phi}_0)=[-\pi,\pi]^m$ is the smallest possible in the sense that with a smaller support the condition (2.8) from \cite{CoGaJe_21}, which is equivalent to the orthonormality, cannot be satisfied.  
\end{example}

\section{Convergence of discretized matrix-valued Fourier multipliers} \label{sec:main}

\begin{assumptions} \label{ass:1}
Let $G_{11}$ and $G_{12}$ be continuously differentiable $d\times d$ matrix-valued functions in $\R^m$ such that the following conditions are satisfied.
\begin{enumerate}[label=(\roman*)]
\item $(\forall\xi\in\R^m) (G_{11}(\xi)=G_{11}(\xi)^*)$
\item $(\forall\xi\in\R^m)(G_{12}(\xi)\text{ is normal})$
\item $(\forall\xi\in\R^m)([G_{11}(\xi),G_{12}(\xi)]=0)$
\item $(\exists C,K,\gamma>0)(\forall\xi\in\R^m:\, |\xi|>K)(G_{11}(\xi)^2+G_{12}(\xi)G_{12}(\xi)^*\geq C|\xi|^{2\gamma}I_{\C^d})$
\item $(\exists C>0,\beta\in\R)(\forall \star\in\{11,12\},\,i,j\in\{1,\ldots,d\},\, \xi\in\R^m)(|\nabla(G_{\star}(\xi))_{ij}|\leq C|\xi|^\beta)$.
\end{enumerate}
\end{assumptions}

\begin{remark}
Using 
\begin{equation*}
(G_\star(\xi))_{ij}=(G_\star(0))_{ij}+\int_0^1\xi\cdot\nabla(G_\star(t\xi))_{ij} \dd t
\end{equation*}
we conclude from (v) that every matrix element of $G_{11}(\xi)^2+G_{12}(\xi)G_{12}(\xi)^*$ is bounded by $C+K|\xi|^{2(1+\beta)}$. Consequently, we get
\begin{equation*}
G_{11}(\xi)^2+G_{12}(\xi)G_{12}(\xi)^*\leq (C+K|\xi|^{2(1+\beta)}) I_{\C^d}.
\end{equation*}
Combined with (iv) this yields the following necessary condition 
\begin{equation} \label{eq:beta_restriction}
\gamma\leq 1+\beta.
\end{equation}
In particular, since $\gamma>0$, we always have $\beta>-1$. Later, we will use the following immediate consequence of \eqref{eq:beta_restriction},
\begin{equation} \label{eq:beta_restriction2}
\gamma\geq 2\gamma-\max\{\beta,0\}-1.
\end{equation}
\end{remark}

\begin{remark}
The condition (ii) may be dropped, if one assumes that the lower bound (iv) holds true also for $G_{11}(\xi)^2+G_{12}(\xi)^*G_{12}(\xi)$. All the proofs in this section then work after an obvious little modification.
\end{remark}

Due to (i), the $2d\times 2d$ matrix valued function
\begin{equation*}
 G_0=G_0(\xi):=\begin{pmatrix}
 G_{11}(\xi) & G_{12}(\xi)\\
 G_{12}(\xi)^* & -G_{11}(\xi)
 \end{pmatrix}
\end{equation*}
is hermitian at every $\xi\in\R^m$. Moreover, on $\R^m$,
\begin{equation*}
G_0^2
=\begin{pmatrix}
G_{11}^2+G_{12}G_{12}^* & [G_{11},G_{12}]\\
[G_{12}^*,G_{11}] & G_{11}^2+G_{12}^*G_{12}
\end{pmatrix}
=\begin{pmatrix}
G_{11}^2+G_{12}G_{12}^* & 0\\
0 & G_{11}^2+G_{12}G_{12}^*
\end{pmatrix},
\end{equation*}
because of (ii ), (iii), and the fact that (iii) combined with (i) implies $[G_{12}^*,G_{11}]=0$ on $\R^m$. Using this observation together with (iv), we get
\begin{equation} \label{eq:G_0_bound}
(G_0(\xi)^2+1)^{-1}\leq(1+C|\xi|^{2\gamma})^{-1}I_{\C^{2d}}\quad (\forall\xi\in\R^m:\, |\xi|>K).
\end{equation}
Since $(G_0(\xi)^2+1)^{-1}\leq I_{\C^{2d}}$ for all $\xi\in\R^m$, there exists $C>0$ such that
\begin{equation} \label{eq:G_0_res_est}
(G_0(\xi)^2+1)^{-1}\leq C(1+|\xi|^{2\gamma})^{-1}I_{\C^{2d}}\quad (\forall\xi\in\R^m)
\end{equation}

With the symbol $G_0$ we associate the Fourier multiplier
\begin{equation} \label{eq:Hh0_def}
H_0:=\F^*G_0\F,
\end{equation}
where $G_0$ is understood as the multiplication operator  on the maximal domain $\Dom{G_0}=\{f\in\hat\HH^{2d}:\, \xi\mapsto G_0(\xi)f(\xi)\in\hat\HH^{2d}\}$. Since $G_0(\xi)$ is hermitian at every $\xi\in\R^m$, the operator $G_0$ is self-adjoint in $\hat\HH^{2d}$. Consequently,  $H_0$ is self-adjoint in $\HH^{2d}$. For the operator norm of its resolvent we have
\begin{multline*}
\|(H_0\pm\iu)^{-1}\|^2=\|(G_0\pm i)^{-1}\|^2\\
=\sup_{\xi\in\R^m}\|(G_0(\xi)\pm i)^{-1}\|^2_{\B(\C^{2d})}=\sup_{\xi\in\R^m}\|(G_0(\xi)^2+1)^{-1}\|_{\B(\C^{2d})},
\end{multline*}
where we used the $C^*$-identity in $\B(\C^{2d})$, due to which for every $A=A^*\in\B(\C^{2d})$, 
\begin{equation} \label{eq:Cstar}
\|(A\pm\iu)^{-1}\|^2_{\B(\C^{2d})}=\|(A^2+1)^{-1}\|_{\B(\C^{2d})}.
\end{equation}

Now, let us introduce two discretized versions of the symbol $G_0$. The first one is motivated by the fact that the symmetric difference in $j$-th variable acts as the multiplication by $\iu\sin(h\xi_j)/h$ in $\hHh$. Therefore, we put
\begin{equation*} 
G_h(\xi):=G_0\Big(\frac{\sin(h\xi_1)}{h},\ldots, \frac{\sin(h\xi_m)}{h}\Big).
\end{equation*}
We will denote the blocks of $G_h$ by $G_{\star,h},\, \star\in\{11,12\}$.
The second version is derived from the functional calculus for the minus Laplacian and its standard discretized version, cf. \cite{CoGaJe_21}, where the same discretization was used, for a more detailed explanation. Since the second derivatives in the standard discrete Laplacian act as \eqref{eq:std_2nd_diff} in the frequency domain, it seems reasonable to substitute $\frac{2}{h}|\sin(h\xi_j/2)|$ for $\xi_j$ in the arguments of $G_0$. However, as has been explained in Introduction, this may lead to a good approximation of $G_0$ only when the latter obeys
\begin{equation} \label{eq:Gsymm}
G_0(\xi)=G_0(|\xi_1|,\ldots,|\xi_m|).
\end{equation}
If this is the case, we put
\begin{equation} \label{eq:tildeGh}
\tilde G_h(\xi):=G_0\Big(\frac{2}{h}\sin\Big(\frac{h}{2}\xi_1\Big),\ldots, \frac{2}{h}\sin\Big(\frac{h}{2}\xi_m\Big)\Big).
\end{equation}
Note that $\tilde G_h$ is a $(2\pi/h)^m$-periodic function, and so it is a well-behaved symbol in $\hHh^{2d}$.

It is easy to check that the matrix-valued functions $G_h$ and $\tilde G_h$ are hermitian at every point and bounded. Consequently, the corresponding multiplication operators have the same properties and are defined on the whole space $\hHh^{2d}$. Using the discrete Fourier transform, we introduce the corresponding Fourier multipliers as
\begin{equation} \label{eq:Hh_def}
H_h:=\Fh^* G_h\Fh,\quad \tilde H_h:=\Fh^*\tilde G_h\Fh.
\end{equation}
The operators $H_h$ and $\tilde H_h$ are everywhere defined, bounded, and self-adjoint in $\Hh^{2d}$. When dealing with the resolvents of these operators, estimates analogous to \eqref{eq:G_0_res_est} will prove useful. We will need them to hold on $[-\frac{3\pi}{2h},\frac{3\pi}{2h}]^m$ with every $h>0$, 
%(or any other cube that contains $[-\frac{\pi}{h},\frac{\pi}{h}]^m$)
which is not possible for $G_h$ because
$G_h(\pm\frac{\pi}{h},\ldots,\pm\frac{\pi}{h})=G_0(0,\ldots,0)$. To overcome this difficulty, we will add a correction term to $G_h$.  Note that the same problem does not occur for $\tilde G_h$ due to the presence of $1/2$ in the arguments of the sine functions that appear in its definition. 

Motivated by \cite{CoGaJe_23}, we put
\begin{equation*} 
G_h^+:=G_h+
\begin{pmatrix}
f_h I_{\C^d} & 0\\
0 & -f_h I_{\C^d}
\end{pmatrix},\qquad
f_h=f_h(\xi):=\Big(h\sum_{j=1}^m\frac{4}{h^2}\sin^2\Big(\frac{h\xi_j}{2}\Big)\Big)^\gamma.
\end{equation*}
The symbol $f_h$ corresponds to $(-h\Delta_h)^\gamma$. We will also use the following notation,
\begin{equation} \label{eq:Hh+_def}
H_h^+:=\Fh^*G_h^+\Fh=H_h+(-h\Delta_h)^\gamma\begin{pmatrix}I_{\C^d} & 0\\ 0&-I_{\C^d}\end{pmatrix}.
\end{equation}
Although $\lim_{h\to 0+}\|f_h\|_{\mathscr{B}(\hHh)}=+\infty$, the correction term is small relatively to $G_h$ in the sense that the difference of the resolvents of $G_h^+$ and $G_h$ tends to zero in the strong operator topology as $h\to 0+$,  which follows from Theorem \ref{theo:strong_conv}. Whenever dealing with $G_h$ or $G_h^+$ we will always assume that 
\begin{equation} \label{eq:G_11_pos}
G_{11}(\xi)\geq 0 \quad (\forall\xi\in\R^m).
\end{equation}
The reason for this restriction is apparent from the proof of the following lemma. Before stating it, we summarize the additional assumptions we have made on the considered symbols.
\begin{assumptions} \label{ass:2}
When dealing with $\tilde{G}_h$ we assume that $G_0$ obeys \eqref{eq:Gsymm}. When dealing with $G_h$ or $G_h^+$ we assume that $G_0$ obeys \eqref{eq:G_11_pos}.
\end{assumptions}

\begin{lemma}
There exists $C>0$ such that for all $h>0$ and $h\xi\in[-3\pi/2,3\pi/2]^m$,
\begin{align}
&(\tilde G_h(\xi)^2+1)^{-1}\leq C(1+|\xi|^{2\gamma})^{-1}I_{\C^{2d}}, \label{eq:tildeG_res_est}\\ &(G_h^+(\xi)^2+1)^{-1}\leq C(1+|\xi|^{2\gamma})^{-1}I_{\C^{2d}}.\label{eq:G+_res_est} 
\end{align}
\end{lemma}

\begin{proof}
To get the first inequality, we note that there exists $C_1>0$ such that  $|\sin(h\xi_j/2)|\geq C_1 h|\xi_j|$ for all $h\xi_j\in[-3\pi/2,3\pi/2]$. Then we just combine \eqref{eq:tildeGh} and \eqref{eq:G_0_res_est}. For the second inequality, we derive from \eqref{eq:G_11_pos} and non-negativity of $f_h(\xi)$ that
\begin{multline*}
G_h^+(\xi)^2=\big((G_{11,h}(\xi)+f_h(\xi) I_{\C^d})^2+G_{12,h}(\xi)G_{12,h}(\xi)^*\big)\otimes I_{\C^2}\\
\geq \big(G_{11,h}(\xi)^2+G_{12,h}(\xi)G_{12,h}(\xi)^*\big)\otimes I_{\C^2}=G_0\Big(\frac{\sin(h\xi_1)}{h},\ldots, \frac{\sin(h\xi_m)}{h}\Big)^2.
\end{multline*} 
This, \eqref{eq:G_0_bound}, and the estimate $|\sin(h\xi_j)|\geq Kh|\xi_j|$, which is valid with some $K>0$ for all $h\xi_j\in[-3\pi/4,3\pi/4]$, yield
\begin{multline*}
(G_h^+(\xi)^2+1)^{-1}\leq \Big(G_0\Big(\frac{\sin(h\xi_1)}{h},\ldots, \frac{\sin(h\xi_m)}{h}\Big)^2+1\Big)^{-1}\\
\leq C\Big(1+\Big|\Big(\frac{\sin(h\xi_1)}{h},\ldots, \frac{\sin(h\xi_m)}{h}\Big)\Big|^{2\gamma}\Big)^{-1}I_{\C^{2d}}\leq C(1+|\xi|^{2\gamma})^{-1}I_{\C^{2d}},\\
\end{multline*}
i.e. \eqref{eq:G+_res_est} for all $h\xi\in[-3\pi/4,3\pi/4]^m$. For $h\xi\in[-3\pi/2,3\pi/2]^m\setminus[-3\pi/4,3\pi/4]^m$, we have
\begin{equation*}
G_h^+(\xi)^2\geq f_h(\xi)^2 I_{\C^{2d}}\geq (4C_1^2 h|\xi|^2)^{2\gamma} I_{\C^{2d}}\geq(3\pi C_1^2|\xi|)^{2\gamma} I_{\C^{2d}},
\end{equation*}
which leads to \eqref{eq:G+_res_est} on the rest of $[-3\pi/2,3\pi/2]^m$.
\end{proof}

We are now ready to show a fundamental result on a uniform bound for the point-wise resolvents. Its proof essentially mimics the proof of \cite[Lemma 3.4]{CoGaJe_23}. Nevertheless, we slightly modified the argument to weaken one of the restrictions on parameters (namely, $2\gamma\leq \beta+3$ in our notation) that appears there.

\begin{proposition} \label{prop:pointwise_res_est}
If $2\gamma>\max\{\beta,0\}+1$, then there exists $C>0$ such that for all $h>0$ and $h\xi\in[-3\pi/2,3\pi/2]^m$,
\begin{align} 
&\|(\tilde G_h(\xi)\pm\iu)^{-1}-(G_0(\xi)\pm\iu)^{-1}\|_{\B(\C^{2d})}\leq Ch^{\min\{2, 2\gamma-\max\{\beta,0\}-1\}}, \label{eq:pointwise_res_est_tilde}
\\
&\|(G_h^+(\xi)\pm\iu)^{-1}-(G_0(\xi)\pm\iu)^{-1}\|_{\B(\C^{2d})}\leq Ch^{\min\{2, 2\gamma-\max\{\beta,0\}-1\}}.  \label{eq:pointwise_res_est_+}
\end{align}
\end{proposition}
\begin{proof}
We will first prove the first inequality and then we will briefly explain what is different when showing the other one. We start with
\begin{equation} \label{eq:2nd_res_form}
 (\tilde G_h(\xi)\pm\iu)^{-1}-(G_0(\xi)\pm\iu)^{-1}=(\tilde G_h(\xi)\pm\iu)^{-1}(G_0(\xi)-\tilde G_h(\xi))(G_0(\xi)\pm\iu)^{-1}.
\end{equation}
The norms of the first and the third term on the right-hand side may be controlled using \eqref{eq:Cstar}, \eqref{eq:tildeG_res_est}, and \eqref{eq:G_0_res_est} as follows
\begin{equation} \label{eq:matrix_res_est}
\|(\tilde G_h(\xi)\pm\iu)^{-1}\|_{\B(\C^{2d})}\leq C(1+|\xi|^{2\gamma})^{-1/2},\,\, \|(G_0(\xi)\pm\iu)^{-1}\|_{\B(\C^{2d})}\leq C(1+|\xi|^{2\gamma})^{-1/2}.
\end{equation}
To control the norm of the middle term, it is sufficient to find a bound for all its matrix elements. Given $\star\in\{11,12\}$ and $i,j\in\{1,\ldots,d\}$, put 
$g_0(\xi):=(G_\star(\xi))_{ij}$ and $g_h(\xi):=(\tilde G_{\star,h}(\xi))_{ij}$.
Then we have
\begin{equation*}
g_h(\xi)=g_0\Big(\frac{2}{h}\sin\Big(\frac{h}{2}\xi_1\Big),\ldots, \frac{2}{h}\sin\Big(\frac{h}{2}\xi_m\Big)\Big).
\end{equation*}
Using Taylor's formula with remainder and the fundamental theorem of calculus we get
\begin{equation*}
\frac{2}{h}\sin\Big(\frac{h}{2}\xi_j\Big)=\xi_j-\frac{h^2}{8}\xi_j^3\int_0^1\cos(th\xi_j/2)(1-t)^2 \dd t=:\xi_j+\zeta_{h,j}
\end{equation*}
and
\begin{equation*}
g_h(\xi)=g_0(\xi_1+\zeta_{h,1},\ldots, \xi_m+\zeta_{h,m})=g_0(\xi)+\int_0^1\zeta_h\cdot\nabla g_0(\xi+t\zeta_h)\dd t,
\end{equation*}
respectively. 

\emph{The case $\beta\geq 0$}. With the help of $|\zeta_h|\leq C h^2|\xi|^3$, (v), and $h|\xi_j|\leq 3\pi/2$, we arrive at
\begin{multline} \label{eq:g_h-g_0}
|g_h(\xi)-g_0(\xi)|\leq|\zeta_h|\int_0^1|\nabla g_0(\xi+t\zeta_h)|\dd t\leq |\zeta_h|\int_0^1|\xi+t\zeta_h|^\beta\dd t\\
\leq C h^2|\xi|^3(|\xi|^\beta+h^{2\beta}|\xi|^{3\beta})\leq C h^2|\xi|^{3+\beta}
\end{multline}
for all $h\xi\in[-3\pi/2,3\pi/2]^m$.

Putting this, \eqref{eq:matrix_res_est} and \eqref{eq:2nd_res_form} together, we obtain
\begin{equation*}
\|(\tilde G_h(\xi)\pm\iu)^{-1}-(G_0(\xi)\pm\iu)^{-1}\|_{\B(\C^{2d})}
\leq C\frac{ h^2|\xi|^{3+\beta}}{1+|\xi|^{2\gamma}} \leq C h^\delta\,\frac{ |\xi|^{1+\beta+\delta}}{1+|\xi|^{2\gamma}}\leq Ch^\delta,
\end{equation*}
where $\delta=\min\{2, 2\gamma-\beta-1\}$.
This shows \eqref{eq:pointwise_res_est_tilde} for $\beta\geq 0$. 

\emph{The case $\beta<0$}. Combining (v) and $g_0\in C^1(\R^m)$, we get $\nabla g_0\leq C$. By a similar reasoning as in the previous case, we arrive at
\begin{equation*}
|g_h(\xi)-g_0(\xi)|\leq|\zeta_h|\int_0^1|\nabla g_0(\xi+t\zeta_h)|\dd t\leq C h^2|\xi|^3  \quad (\forall\, h\xi\in[-3\pi/2,3\pi/2]^m),
\end{equation*}
which is \eqref{eq:g_h-g_0} with $\beta=0$. Consequently,
\begin{equation*}
\|(\tilde G_h(\xi)\pm\iu)^{-1}-(G_0(\xi)\pm\iu)^{-1}\|_{\B(\C^{2d})}
\leq Ch^\delta,
\end{equation*}
with $\delta=\min\{2, 2\gamma-1\}$. This implies \eqref{eq:pointwise_res_est_tilde} for $\beta< 0$. 

The proof of \eqref{eq:pointwise_res_est_+} essentially differs only at the point where the elements of the middle term of the right-hand side of \eqref{eq:2nd_res_form} (with $G_h^+$ instead of $G_h$) are estimated. With the notation $g_h^+(\xi):=(G_{11,h}^+(\xi))_{jj},\, j=1,\ldots, d$, and $\beta^+:=\max\{\beta, 0\}$ we have
\begin{equation*}
|g_h^+(\xi)-g_0(\xi)|\leq C h^2|\xi|^{3+\beta^+}+|f_h(\xi)|\\
\leq C h^2|\xi|^{3+\beta^+}+h^\gamma|\xi|^{2\gamma}.
\end{equation*}
The remaining elements that do not contain the correction term are estimated just by $C h^2|\xi|^{3+\beta^+}$. We conclude that
\begin{multline*}
\|(G_h^+(\xi)\pm\iu)^{-1}-(G_0(\xi)\pm\iu)^{-1}\|_{\B(\C^{2d})}
\leq C h^\delta\,\frac{ |\xi|^{1+\beta^+ +\delta}}{1+|\xi|^{2\gamma}}+Ch^\gamma\frac{ |\xi|^{2\gamma}}{1+|\xi|^{2\gamma}}\\
\leq Ch^\delta\frac{ |\xi|^{1+\beta^+ +\delta}+|\xi|^{\gamma+\delta}}{1+|\xi|^{2\gamma}}\leq C h^\delta,
\end{multline*}
where $\delta=\min\{2, \gamma, 2\gamma-\beta^+-1\}$. Finally, due to \eqref{eq:beta_restriction2}, $\delta=\{2,2\gamma-\beta^+-1\}$.
\end{proof}

\begin{remark}
Note that the key assumption $2\gamma>\max\{\beta,0\}+1$ may be only satisfied for $\gamma>1/2$. In view of \eqref{eq:beta_restriction}, this yields further restriction on $\beta$, namely $\beta>-1/2$ .
\end{remark}

\begin{remark} \label{rem:comparison}
Although we dropped one assumption  of \cite[Lemma 3.4]{CoGaJe_21}, which in our notation reads 
\begin{equation} \label{eq:2gamma_beta_3}
2\gamma\leq \beta+3,
\end{equation}
the bound \eqref{eq:pointwise_res_est_tilde}  is in the end equivalent to that of \cite[Lemma 3.4]{CoGaJe_21}, that deals with scalar symbols.  To see it, consider the case when $\beta\geq 0$ (for $\beta<0$ one may proceed similarly), and look at the exponent of the convergence rate, i.e. $\min\{2, 2\gamma-\beta-1\}$. The condition $2\geq 2\gamma-\beta-1$ is equivalent to \eqref{eq:2gamma_beta_3}. If the opposite is true, the exponent in the convergence rate is $2$. But then we may take $\tilde\gamma<\gamma$ such that $2\tilde\gamma=\beta+3$ instead of $\gamma$. It will always satisfy (iv). Therefore, assuming \eqref{eq:2gamma_beta_3}, the upper bounds in Proposition \ref{prop:pointwise_res_est} would be just $C h^{2\gamma-\max\{\beta, 0\}-1}$. The same observation applies to the Theorem \ref{theo:main} below.
\end{remark}

The principal result of this section reads as follows.

\begin{theorem} \label{theo:main}
Let $G_0$ satisfy Assumptions \ref{ass:1} and \ref{ass:2}, $H_0,\tilde H_h,$ and $H_h^+$ be the Fourier multipliers whose symbols are defined in \eqref{eq:Hh0_def}, \eqref{eq:Hh_def}, and \eqref{eq:Hh+_def}, respectively.
If $2\gamma>\max\{\beta,0\}+1$  then there exist $C>0$ and $h_0>0$ such that for all $h\in(0,h_0)$, 
\begin{align} 
&\|J_h(\tilde H_h\pm\iu)^{-1}K_h-(H_0\pm\iu)^{-1}\|\leq Ch^{\min\{2, 2\gamma-\max\{\beta,0\}-1\}}, \label{eq:res_est_tilde}
\\
&\|J_h (H_h^+\pm\iu)^{-1}K_h-(H_0\pm\iu)^{-1}\|\leq Ch^{\min\{2, 2\gamma-\max\{\beta,0\}-1\}}.  \label{eq:res_est_+}
\end{align}
\end{theorem}

\begin{proof}
One can essentially follow step by step the proof of \cite[Proposition 3.5]{CoGaJe_21}. Therefore, we will just summarize its main ideas for the reader's convenience. Let us do it for the first bound. 

By the triangle inequality,
\begin{multline} \label{eq:res_triangle}
\|J_h(\tilde H_h\pm\iu)^{-1}K_h-(H_0\pm\iu)^{-1}\|\\
\leq\|J_h(\tilde H_h\pm\iu)^{-1}K_h-J_h K_h(H_0\pm\iu)^{-1}\|+\|(I_{\HH^{2d}}-J_h K_h)(H_0\pm\iu)^{-1}\|.
\end{multline}
To estimate the latter term, it suffices to realize that given $L>0$, \eqref{eq:G_0_bound} and \eqref{eq:Cstar} imply that there exists $h_0$ such that for all $h\in(0,h_0)$,
\begin{equation} \label{eq:G_0_h_bound}
\|(G_0(\xi)\pm\iu)^{-1}\|_{\B(\C^{2d})}\leq\frac{1}{(1+C|\xi|^{2\gamma})^{1/2}}\leq\frac{C}{|\xi|^\gamma}\leq Ch^\gamma\quad (\forall\xi\in\R^m:\,  h|\xi|>L),
\end{equation}
and then modify the proof of \cite[Lemma 3.3]{CoGaJe_21} for matrix-valued symbols in an obvious way to arrive at 
\begin{equation*}
\|(I_{\HH^{2d}}-J_h K_h)(H_0\pm\iu)^{-1}\|\leq Ch^\gamma \quad (\forall h\in(0,h_0)).
\end{equation*}

The action of the first term in the bound \eqref{eq:res_triangle} on every $g$ in the space $\mathcal{S}(\R^m;\C^{2d})$ (of $\C^{2d}$-valued functions whose components belong to the Schwartz space) can be written in the frequency domain as
\begin{equation*}
 \F\big(J_h(\tilde H_h\pm\iu)^{-1}K_h-J_h K_h(H_0\pm\iu)^{-1}\big)\F^*g=\sum_{j\in\Z^m}q_{j,h},
\end{equation*}
where 
\begin{equation} \label{eq:q_def}
q_{j,h}(\xi):=(2\pi)^m\hat\phi_0(h\xi)\overline{\hat\psi_0(h\xi+2\pi j)}S_h(\xi+2\pi j/h)g(\xi+2\pi j/h) \quad (\forall \xi\in\R^m)
\end{equation}
with $S_h(\xi):=(\tilde G_h(\xi)\pm\iu)^{-1}-(G_0(\xi)\pm\iu)^{-1}$, see formula (3.21) in \cite{CoGaJe_21} that remains valid also in our setting.
Due to the support conditions in \eqref{eq:supports_ass}, only $q_{j,h}(\xi)$ with $|j|\leq 1$ and $h\xi\in[-3\pi/2,3\pi/2]^m$ is possibly non-zero. When $j=0$, we use \eqref{eq:pointwise_res_est_tilde} to control $S_h$. Taking into account the fact that $\hat\phi_0$ and $\hat\psi_0$ are bounded, we get
\begin{equation} \label{eq:q_0_bound}
|q_{0,h}(\xi)|\leq C h^{\min\{2,2\gamma-\max\{\beta,0\}-1\}} |g(\xi)|.
\end{equation} 
If $|j|=1$, then it is sufficient (again, because of \eqref{eq:supports_ass}) to consider $\xi$ such that $h\xi\in[-3\pi/2,3\pi/2]^m$ and $h|\xi|\geq\pi/2$, for which we write
\begin{multline*}
|q_{j,h}(\xi)|\\
\leq C\big(\|(\tilde G_h(\xi+2\pi j/h)\pm\iu)^{-1}\|_{\B(\C^{2d})}+\|(G_0(\xi+2\pi j/h)\pm\iu)^{-1}\|_{\B(\C^{2d})}\big)|g(\xi+2\pi j/h)|.
\end{multline*}
Since $|h\xi+2\pi j|\geq \pi/2$, we can estimate $\|(G_0(\xi+2\pi j/h)\pm\iu)^{-1}\|_{\B(\C^{2d})}$ as in \eqref{eq:G_0_h_bound}. Furthermore, by periodicity of $\tilde G_h$ and \eqref{eq:matrix_res_est}, we obtain
\begin{multline*} 
\|(\tilde G_h(\xi+2\pi j/h)\pm\iu)^{-1}\|_{\B(\C^{2d})}=\|(\tilde G_h(\xi)\pm\iu)^{-1}\|_{\B(\C^{2d})}\leq \frac{C}{(1+|\xi|^{2\gamma})^{1/2}}\\
\leq \frac{C}{|\xi|^\gamma}\leq C h^\gamma.
\end{multline*}
We conclude that 
\begin{equation*}
|q_{j,h}(\xi)|\leq C h^\gamma|g(\xi+2\pi j/h)| \quad (\forall j:\, |j|=1).
\end{equation*}
Integrating the square of this and that of \eqref{eq:q_0_bound} over $\R^m$, we arrive at
\begin{multline*}
\|\F\big(J_h(\tilde H_h\pm\iu)^{-1}K_h-J_h K_h(H_0\pm\iu)^{-1}\big)\F^*g\|_{\hat \HH^{2d}}\\
\leq C h^{\min\{2,\gamma,2\gamma-\max\{\beta,0\}-1\}}\|g\|_{\hat \HH^{2d}},
\end{multline*}
from which \eqref{eq:res_est_tilde} with $Ch^{\min\{2, \gamma, 2\gamma-\max\{\beta,0\}-1\}}$ on the right-hand side immediately follows. Finally, we use \eqref{eq:beta_restriction2}.

The proof of \eqref{eq:res_est_+} is analogous.
\end{proof}

\begin{remark}
Employing the same strategy as above, one cannot improve the convergence rate in \eqref{eq:res_est_+} by choosing a different correction term, because the used correction term is controlled by the same power of $h$ as $\|(G_0(\xi+2\pi j/h)\pm\iu)^{-1}\|_{\B(\C^{2d})}$. 
\end{remark}

When the correction term is absent, $H_h$ itself does not converge to $H_0$ in the \emph{generalized norm} resolvent sense.

\begin{proposition} \label{prop:neg_result}
There exist $C>0$ such that for all $h\in(0,1)$,
\begin{equation} \label{eq:pointwise_divergence}
\sup_{\xi\in[-\pi/h,\pi/h]^m}\|(G_h(\xi)\pm\iu)^{-1}-(G_h^+(\xi)\pm\iu)^{-1}\|_{\B(\C^{2d})}\geq C.
\end{equation}
Consequently, under the assumptions of Theorem \ref{theo:main}, we have
\begin{equation} \label{eq:div_result}
\lim_{h\to 0+}\|J_h (H_h\pm\iu)^{-1}K_h-(H_0\pm\iu)^{-1}\|\neq 0.
\end{equation}
\end{proposition}

\begin{proof}
Using the second resolvent identity, we write
\begin{equation*}
 (G_h(\xi)\pm\iu)^{-1}-(G_h^+(\xi)\pm\iu)^{-1}=A(\xi)B(\xi)D(\xi),
\end{equation*}
where 
\begin{equation*}
A(\xi):=(G_h(\xi)\pm\iu)^{-1},\quad 
B(\xi):=f_h(\xi) \begin{pmatrix} I_{\C^d} & 0\\ 0 & -I_{\C^d}\end{pmatrix},
\quad D(\xi):=(G_h^+(\xi)\pm\iu)^{-1}.
\end{equation*}
Since $A(\xi)$ and $D(\xi)$ are invertible, we get 
\begin{equation} \label{eq:ABD}
\|A(\xi)B(\xi)D(\xi)\|_{\B(\C^{2d})}\geq\frac{\|B(\xi)\|_{\B(\C^{2d})}}{\|A(\xi)^{-1}\|_{\B(\C^{2d})}\|D(\xi)^{-1}\|_{\B(\C^{2d})}}. 
\end{equation}
We will evaluate the norms on the right-hand side at $\xi_j=\pi/h,\, \forall j\in\{1,\ldots,m\}$. First, we have
\begin{equation} \label{eq:A_bound}
\|A(\pi/h,\ldots,\pi/h)^{-1}\|_{\B(\C^{2d})}=\|G_0(0)\pm\iu\|_{\B(\C^{2d})}=\|G_0(0)^2+1\|_{\B(\C^{2d})}^{1/2},
\end{equation}
where $0$ in the argument of $G_0$ stands for the null vector in $\R^m$. Similarly, we get
\begin{multline*}
\|D(\pi/h,\ldots,\pi/h)^{-1}\|_{\B(\C^{2d})}\\
=\Big\|G_0(0)+f_h(\pi/h,\ldots,\pi/h)\begin{pmatrix} I_{\C^d} & 0\\ 0 & -I_{\C^d}\end{pmatrix}\,\pm\iu\,\Big\|_{\B(\C^{2d})}\\
=\|(G_{11}(0)+(4m/h)^\gamma)^2+G_{12}(0)G_{12}(0)^*+1\|_{\B(\C^{d})}^{1/2}.
\end{multline*}
Inserting this, \eqref{eq:A_bound}, and $\|B(\pi/h,\ldots,\pi/h)\|_{\B(\C^{2d})}=(4m/h)^\gamma$ to \eqref{eq:ABD} we arrive at
\begin{multline*}
\|(ABD)(\pi/h,\ldots,\pi/h)\|_{\B(\C^{2d})} \\
\geq \frac{(4m)^\gamma}{\|G_0(0)^2+1\|_{\B(\C^{2d})}^{1/2}
\|(h^\gamma G_{11}(0)+(4m)^\gamma)^2+ h^{2\gamma}(G_{12}(0)G_{12}(0)^*+1)\|_{\B(\C^{d})}^{1/2}} \\
\geq \frac{(4m)^\gamma}{\|G_0(0)^2+1\|_{\B(\C^{2d})}^{1/2}
\|(G_{11}(0)+(4m)^\gamma)^2+ G_{12}(0)G_{12}(0)^*+1\|_{\B(\C^{d})}^{1/2}},
\end{multline*}
where the last estimate holds for every $h\in(0,1)$. This shows \eqref{eq:pointwise_divergence}.

Now, recall $K_h J_h=I_{\Hh^{2d}}$, and that both $K_h$ and $J_h$ are bounded, uniformly in $h>0$. Therefore, for every $E\in\B(\Hh^{2d})$, we get $\|E\|=\|K_h J_h E K_h J_h\|\leq C\|J_h E K_h\|$, from which it follows that $\|J_h E K_h\|\geq C\|E\|$. Choosing 
$$E=\Fh^*((G_h\pm\iu)^{-1}-(G_h^+\pm\iu)^{-1})\Fh=(H_h\pm\iu)^{-1}-(H_h^+\pm\iu)^{-1}$$
and taking into account unitarity of $\Fh$ jointly with \eqref{eq:pointwise_divergence}, we get 
\begin{equation*}
\|J_h((H_h\pm\iu)^{-1}-(H_h^+\pm\iu)^{-1})K_h\|\geq C
\end{equation*}
for all $h\in (0,1)$. Putting this together with \eqref{eq:res_est_+} yields \eqref{eq:div_result}.
\end{proof}

On the other hand, $H_h$  still converges to $H_0$ in the \emph{generalized strong} resolvent sense, even when the assumption \eqref{eq:G_11_pos} is dropped. 

\begin{theorem} \label{theo:strong_conv}
There exists $C>0$ such that for $h>0$ and $h\xi\in[-3\pi/2,3\pi/2]^m$, we have
\begin{equation} \label{eq:pointwise_strong}
\|((G_h(\xi)\pm\iu)^{-1}-(G_h^+(\xi)\pm\iu)^{-1})(G_0(\xi)\pm\iu)^{-1}\|_{\B(\C^{2d})}\leq Ch^\gamma.
\end{equation}
Consequently, if $2\gamma>\max\{\beta,0\}+1$, then there exist $C>0$ and $h_0>0$ such that for all $h\in(0,h_0)$, 
\begin{equation}  \label{eq:strong_conv_est}
\|(J_h(H_h\pm\iu)^{-1}K_h-(H_0\pm\iu)^{-1})(H_0\pm\iu)^{-1}\|\leq Ch^{\min\{2, 2\gamma-\max\{\beta,0\}-1\}}.
\end{equation}
In particular, for every $\psi\in\HH^{2d}$,
\begin{equation} \label{eq:strong_res_conv}
\lim_{h\to 0+}\|(J_h(H_h\pm\iu)^{-1}K_h-(H_0\pm\iu)^{-1})\psi\|=0.
\end{equation}
\end{theorem}

\begin{proof}
Using the same notation as in the proof of Proposition \ref{prop:neg_result}, 
\begin{multline*}
\|((G_h(\xi)\pm\iu)^{-1}-(G_h^+(\xi)\pm\iu)^{-1})(G_0(\xi)\pm\iu)^{-1}\|_{\B(\C^{2d})}\\
\leq \|A(\xi)\|_{\B(\C^{2d})}\|B(\xi)\|_{\B(\C^{2d})}\|D(\xi)\|_{\B(\C^{2d})}\|(G_0(\xi)\pm\iu)^{-1}\|_{\B(\C^{2d})}.
\end{multline*}
Recall that, by \eqref{eq:G_0_res_est} and \eqref{eq:Cstar}, 
\begin{equation*}
\|(G_0(\xi)\pm\iu)^{-1}\|_{\B(\C^{2d})}\leq\frac{C}{(1+|\xi|^{2\gamma})^{1/2}} \quad(\forall\xi\in\R^m),
\end{equation*}
which also implies 
\begin{equation} \label{eq:A_bound2}
\|A(\xi)\|_{\B(\C^{2d})}=\|(G_h(\xi)\pm\iu)^{-1}\|_{\B(\C^{2d})}\leq C \quad(\forall\xi\in\R^m).
\end{equation}
Similarly, from \eqref{eq:G+_res_est} it follows that 
\begin{equation} \label{eq:D_est}
\|D(\xi)\|_{\B(\C^{2d})}=\|(G_h^+(\xi)\pm\iu)^{-1}\|_{\B(\C^{2d})}\leq\frac{C}{(1+|\xi|^{2\gamma})^{1/2}} \quad (\forall \, h\xi\in[-3\pi/2,3\pi/2]^m).
\end{equation}
Finally, since $\|B(\xi)\|_{\B(\C^{2d})}\leq(h|\xi|^2)^\gamma$, we conclude that
\begin{equation*}
\|((G_h(\xi)\pm\iu)^{-1}-(G_h^+(\xi)\pm\iu)^{-1})(G_0(\xi)\pm\iu)^{-1}\|_{\B(\C^{2d})}\leq \frac{C h^\gamma |\xi|^{2\gamma}}{1+|\xi|^{2\gamma}}\leq Ch^\gamma,
\end{equation*}
for all $h\xi\in[-3\pi/2,3\pi/2]^m$, which shows \eqref{eq:pointwise_strong}.

Now, we will prove that
\begin{equation*}
\|J_h((H_h\pm\iu)^{-1}-(H_h^+\pm\iu)^{-1})K_h(H_0\pm\iu)^{-1}\|\leq Ch^\gamma.
\end{equation*} 
This estimate combined with Theorem \ref{theo:main} and boundedness of $(H_0\pm\iu)^{-1}$ immediately imply \eqref{eq:strong_conv_est}. It can be shown by a little modification of the proof of Theorem \ref{theo:main}. Namely, in the expression for $q_{j,h}(\xi)$, see \eqref{eq:q_def}, $S_h(\xi)$ is now given by
\begin{equation*}
S_h(\xi)=((G_h(\xi)\pm\iu)^{-1}-(G_h^+(\xi)\pm\iu)^{-1})(G_0(\xi)\pm\iu)^{-1}.
\end{equation*}
Therefore, with the help of \eqref{eq:pointwise_strong}, we get
$|q_{0,h}(\xi)|\leq Ch^\gamma|g(\xi)|$.
To control $|q_{j,h}(\xi)|$ with $|j|=1$, we use the following estimate,
\begin{multline*}
\|S_h(\xi+2\pi j/h)\|_{\B(\C^{2d})}\\
\leq(\|(G_h(\xi+2\pi j/h)\pm\iu)^{-1}\|_{\B(\C^{2d})}+\|(G_h^+(\xi+2\pi j/h)\pm\iu)^{-1}\|_{\B(\C^{2d})})\\
\times \|(G_0(\xi+2\pi j/h)\pm\iu)^{-1}\|_{\B(\C^{2d})}.
\end{multline*}
Employing \eqref{eq:A_bound2}, periodicity of $G_h^+$ combined with
\eqref{eq:D_est}, and \eqref{eq:G_0_h_bound}, we get the bound
\begin{equation*}
\|S_h(\xi+2\pi j/h)\|_{\B(\C^{2d})}\leq (C+Ch^\gamma)Ch^\gamma\leq Ch^\gamma,
\end{equation*}
valid for all $\xi$ such that $h\xi\in[-3\pi/2,3\pi/2]^m$ and $h|\xi|\geq\pi/2$. The rest is identical to the proof of Theorem \ref{theo:main} but the convergence rate is now just $h^\gamma$.

Finally, to prove \eqref{eq:strong_res_conv}, put
\begin{equation*}
L_h:=J_h(H_h\pm\iu)^{-1}K_h-(H_0\pm\iu)^{-1},\quad \tilde L_h:=L_h(H_0\pm\iu)^{-1}.
\end{equation*}
For every $\psi\in\Dom (H_0)$, we get
\begin{equation*}
 \|L_h\psi\|_{\HH^{2d}}=\|\tilde L_h(H_0\pm i)\psi\|_{\HH^{2d}}\leq \|\tilde L_h\|\|(H_0\pm\iu)\psi\|_{\HH^{2d}}.
\end{equation*}
Consequently, $\lim_{h\to 0+}\|L_h\psi\|=0$. By density, for every $\psi\in\HH^{2d}$ there exists a sequence $(\psi_n)$ such that $\lim_{n\to\infty}\|\psi-\psi_n\|_{\HH^{2d}}=0$. Observing that $L_h$ is uniformly bounded for $h>0$, because
\begin{equation*}
\|L_h\|\leq\|J_h\|\|(H_h\pm\iu)^{-1}\|\|K_h\|+\|(H_0\pm\iu)^{-1}\|\leq C\|(G_h\pm\iu)^{-1}\|+C\leq C,
\end{equation*}
where we used the uniform boundedness of $J_h$ and $K_h$ and \eqref{eq:A_bound2}, we obtain
\begin{multline*}
\|L_h\psi\|_{\HH^{2d}}\leq\|L_h(\psi-\psi_n)\|_{\HH^{2d}}+\|L_h\psi_n\|_{\HH^{2d}}\\
\leq C\|\psi-\psi_n\|_{\HH^{2d}}+\|\tilde L_h\|\|(H_0\pm\iu)\psi_n\|_{\HH^{2d}}.
\end{multline*}
Fixing $n$ and then sending $h\to 0+$, we arrive at \eqref{eq:strong_res_conv}.
\end{proof}

\subsection{Examples}  \label{sec:examples}
Given the mass term $\mass\geq 0$, the symbol $G_0$ for the one-, two-, and three-dimensional \emph{Dirac operator} reads
\begin{equation*}
G_0(\xi)=\begin{pmatrix}
\mass & \xi_1 \\
\xi_1 & -\mass
\end{pmatrix},\quad
\begin{pmatrix}
\mass & \xi_1-\iu\xi_2 \\
\xi_1+\iu\xi_2 & -\mass
\end{pmatrix},\quad
\begin{pmatrix}
\mass I_{\C^2} & \xi\cdot\sigma \\
\xi\cdot\sigma & -\mass I_{\C^2}
\end{pmatrix},
\end{equation*}
respectively. Here, $\xi\cdot\sigma:=\sum_{i=1}^3\xi_i\sigma_i$, where $\sigma_i$ are the usual Pauli matrices. It is straightforward to check that all Assumptions \ref{ass:1} are satisfied with $\gamma=1$ and $\beta=0$. Furthermore, these symbols do not obey \eqref{eq:Gsymm}. Therefore, it makes no sense to use the discretized symbol $\tilde G_h$ given by \eqref{eq:tildeGh}. On the other hand, the convergence result \eqref{eq:res_est_+} for discretized operators derived from the symmetric difference holds true with the convergence rate $\mathcal{O}(h)$. Note that the needed correction term corresponds to replacing $\mass$ by $(\mass-h\Delta_h)$. Exactly the same results were derived before in \cite{CoGaJe_23} by exploring each dimension separately.

Another physically relevant example is the effective \emph{Hamiltonian of bilayer graph\-ene} whose symbol is given by
\begin{equation*}
 G_0(\xi_1,\xi_2)=\begin{pmatrix}
\mass & (\xi_1-\iu\xi_2)^2 \\
(\xi_1+\iu\xi_2)^2 & -\mass
\end{pmatrix},
\end{equation*}
where again $\mass\geq 0$. Assumptions \ref{ass:1} are now satisfied with $\gamma=2$ and $\beta=1$. Using \eqref{eq:res_est_+}, we get that the discretized operators derived from the symmetric difference converge to their continuum counterpart after adding the correction term $(h\Delta_h)^2 \sigma_3$ with the rate $\mathcal{O}(h^2)$. Note that different discretizations derived from backward and forward differences were studied in \cite{Ruben}. 

Let us remark that these two examples are special instances of the so-called \emph{abstract Dirac operator with supersymmetry}\cite[Section 5.4]{Thaller}, where, with the notation
\begin{equation*}
Q:=\begin{pmatrix}
0 & \F^*G_{12}\F\\
\F^*G_{12}^*\F & 0
\end{pmatrix},\quad
\tau=\begin{pmatrix}
I_{\HH^d} & 0\\
0 & -I_{\HH^d}
\end{pmatrix},
\end{equation*}
$Q$ is a supercharge with respect to $\tau$. In fact, any $H_0$ whose symbol satisfies (i), (ii), and (iii) is of that type.

Finally, if we put $d=1$ and $G_{12}=0$, we get $G_0=G_{11}\oplus(-G_{11})$, i.e., studying continuum limits of discretizations of $G_0$ is equivalent to the same problem for $G_{11}$, which is just multiplication by a scalar. The Assumptions \ref{ass:1} are satisfied if and only if $G_{11}$ is a real-valued differentiable function such that
\begin{align*}
&(\exists C,K,\gamma>0)(\forall\xi\in\R^m:\, |\xi|>K)(G_{11}(\xi)^2\geq C|\xi|^{2\gamma}),\\
&(\exists C,\beta>0)(\forall\xi\in\R^m)(|\nabla G_{11}(\xi)|\leq C|\xi|^\beta).
\end{align*}
These conditions are exactly the points (2) and (3) of Assumption 3.1 from \cite{CoGaJe_21}, where \eqref{eq:Gsymm} was imposed in addition. Therein, the convergence of discretized operators given by the symbol \eqref{eq:tildeGh} to the corresponding continuum operator was proven under the additional conditions $2\gamma> \max\{\beta, 0\}+1$ and $2\gamma\leq 3+\beta$ with the convergence rate $\mathcal{O}(h^{2\gamma-\max\{\beta, 0\}-1})$. In view of Remark \ref{rem:comparison}, \cite[Proposition 3.5]{CoGaJe_21} is a special case of  Theorem \ref{theo:main}.

\subsection{Discussion on the correction term} \label{sec:correction}
Recall that the only step, where we needed the correction term, was the proof of the crucial bound \eqref{eq:G+_res_est}. Without the correction term it is not possible to estimate the non-negative matrix $G_h(\xi)^2$ from below by $C|\xi|^\alpha$ for $h\xi\in(-\pi,\pi)^m$ with some $\alpha>0$. For the Dirac operator,  $\xi\mapsto G_h(\xi)^2$ essentially yields the dispersion relation, which then has a minimum not only at $\xi=0$ but also along the boundary of the cell $(-\pi/h,\pi/h)^m$. This effect is known from lattice quantum field theories as infamous Fermionic doubling, cf. also Nielsen-Ninomiya \emph{no-go} theorem for massless fermions on lattice \cite{NiNi_81}. It is typically suppressed by adding a correction (Wilson's) term or by introducing a staggered grid discretization (different components of wave functions live on mutually shifted lattices). See \cite{Na_24} for a mathematical justification of these approaches. Below, we will provide yet another argument why $G_h$ and also another typical discretization of Dirac operators (the "forward-backward model") are not, with a single exception of one dimension, the best possible approximations of their continuum counterparts.

Our direct but rather heuristic reasoning is based on the physical requirement that the square of the Dirac operator should be equal to $(-\Delta+\mass^2)I_{\C^{2d}}$, combined with the empirical fact that the optimal discretization of the Laplacian is formed by the second standard central differences, which are given by the composition of the forward and backward difference operators, see \eqref{eq:disc_laplacian}. 
First, let us look at two commonly used discretized versions of the one-dimensional Dirac operator, namely,
\begin{equation*}
D_h^{\textup{symm}}\equiv H_h=\begin{pmatrix}
\mass & -\iu d_h\\
-\iu d_h & -\mass
\end{pmatrix}, \quad
D_h^{\textup{fb}}=\begin{pmatrix}
\mass & -\iu d_h^+\\
-\iu d_h^- & -\mass
\end{pmatrix}.
\end{equation*}
For their squares, we get $(D_h^{\textup{symm}})^2=(-d_h^2+\mass^2)I_{\C^2}$ and $(D_h^{\textup{fb}})^2=(-\Delta_h+\mass^2)I_{\C^2}$, respectively. Therefore, one would guess, that $D_h^{\textup{fb}}$  approximates the "continuous" Dirac operator better than $D_h^{\textup{symm}}$. Indeed, $D_h^{\textup{fb}}$ converges to its continuum counterpart in the generalized norm resolvent sense as $h\to 0+$, the correction term is not needed as opposed to $D_h^{\textup{symm}}$ \cite{CoGaJe_23}. In two dimensions, we will consider the following discretized Dirac operators,
\begin{align*}
&D_{2D,h}^{\textup{symm}}\equiv H_h=\begin{pmatrix}
\mass & -\iu d_{h,x}-d_{h,y}\\
-\iu d_{h,x}+d_{h,y} & -\mass
\end{pmatrix}, \\
&D_{2D,h}^{\textup{fb}}=\begin{pmatrix}
\mass & -\iu d_{h,x}^+ -d_{h,y}^-\\
-\iu d_{h,x}^- +d_{h,y}^+ & -\mass
\end{pmatrix}.
\end{align*}
A straightforward calculation yields
\begin{equation*}
(D_{2D,h}^{\textup{symm}})^2=(-d_{h,x}^2-d_{h,y}^2+\mass^2)I_{\C^2}, \, (D_{2D,h}^{\textup{fb}})^2=(-\Delta_h+\mass^2-\iu d_{h,x}^+ d_{h,y}^- + \iu d_{h,y}^+ d_{h,x}^-)I_{\C^2}.
\end{equation*}
In either case we do not get the optimal square, i.e.  $(-\Delta_h+\mass^2)I_{\C^2}$. The generalized norm resolvent convergence of both discretizations to the "continuous" two-dimensional Dirac operator is only possible after the addition of the correction term \cite{CoGaJe_23}. The three-dimensional setting is similar to the two-dimensional case.

Finally, note that the concrete form of the correction term is not important-only certain abstract properties of its symbol matter. Our choice is motivated by a nice explicit form of the correction term in the coordinate representation.

\section{Convergence of operators on varying Hilbert spaces} \label{sec:convergence}
The aim of this section is to recall several concepts of generalized convergence for a sequence of linear operators acting on varying Hilbert spaces and to show that our convergence results are just their special instances. First, we will present a generalization of the concept of Weidmann which is based on identification of varying Hilbert spaces with a common (\emph{parent}) Hilbert space via isometries \cite{Weidmann,We_84}. Next, we will recall a more recent approach developed by Post that is built on the notion of \emph{quasi-unitary equivalence} (QUE), a qualitative generalization of unitary equivalence \cite{Po_06,Olaf}. When the Hilbert spaces are not varying, the both concepts reduce to the usual norm resolvent convergence. Moreover, they are equivalent, but a loss of convergence speed may occur, when passing from QUE-convergence to generalized Weidmann's convergence \cite[Theorems 1.7 and 2.8]{PostZimmer2022}. To prevent this loss, a modified definition of quasi-unitary equivalence has been introduced very recently by Post and Zimmer in \cite{PoZi_25}. However, we will work with the original definition of QUE-convergence, because in our setting it is equivalent to that of Weidmann including the convergence speed. Finally, we will present yet another remarkable concept introduced by Barker and motivated by the quantum mechanical postulates saying that we do not observe operators nor state vectors but only certain scalar products.

\subsection{Generalized norm resolvent convergence} 

\begin{definition}[generalized Weidmann's convergence]

Let $A_n$ be a self-adjoint operator on a Hilbert space $\Gn$ for $n \in \N \cup \{\infty\}$. We say that the sequence $(A_n)_{n\in\N}$ converges to $A_\infty$ in the \emph{generalized norm resolvent sense of Weidmann} if and only if there exist a Hilbert space $\mathscr{G}$,  called parent space, and for each $n \in \N \cup \{\infty\}$ an isometry $\iota_n : \Gn \to \mathscr{G}$ such that $\lim_{n\to \infty}\delta_n=0$, where
    $$\delta_n:=\|\iota_n (A_n - z)^{-1} \iota^*_n - \iota_{\infty} (A_\infty - z)^{-1} \iota_{\infty}^* \|_{\mathscr{B}(\mathscr{G})}$$
and where $z$  is a common resolvent element $z \in \bigcap_{n \in \N \cup \{\infty\}}\varrho(A_n)$.
For brevity, we write $A_n\underset{\text{Weid}}{\overset{\text{gnrc}}{\to}} A_\infty$ and call $(\delta_n)_n$ the convergence speed. 
    
\end{definition}

\begin{definition}[QUE-convergence]
\label{def:QUE_convergence}

Let $A_n$ be a self-adjoint operator on a Hilbert space $\Gn$ for $n \in \N \cup \{\infty\}$. We say that the sequence $(A_n)_n$ converges to $A_\infty$ in the \emph{QUE-generalized norm resolvent sense} if and only if there exist $z \in \bigcap_{n \in \N \cup \{\infty\}}\varrho(A_n)$ and for each $n \in \N \cup \{\infty\}$ a bounded operator $J_n:\Gn \to \mathscr{G}_\infty$ and $\delta_n\geq 0$ such that 
the following conditions are true:
\begin{gather*}
\lim_{n\to \infty}\delta_n=0,\quad
\|J_n\|_{\mathscr{B}(\Gn, \mathscr{G}_\infty)} \leq 1 + \delta_n,\\
\| (I_{\Gn} - J^*_n J_n)(A_n-z)^{-1}\|_{\mathscr{B}(\Gn)} \leq \delta_n , \quad
\|(I_{\mathscr{G}_\infty} - J_n J^*_n)(A_\infty-z)^{-1}\|_{\mathscr{B}(\mathscr{G}_\infty)} \leq \delta_n,\\
\|(A_\infty-z)^{-1} J_n - J_n (A_n-z)^{-1}\|_{\mathscr{B}(\Gn, \mathscr{G}_\infty)} \leq \delta_n.
\end{gather*}
If $z \in \C \setminus \R,$ we require that the above also holds for $z^*$. For short, we write $A_n\underset{\text{QUE}}{\overset{\text{gnrc}}{\to}} A_\infty$ and call $(\delta_n)_n$ the convergence speed.

\end{definition}

The following result is just a condensed version of \cite[Theorems 2.4, 2.8, and 2.12]{PostZimmer2022}.

\begin{theorem}
 Let $A_n$ be a self-adjoint operator on a Hilbert space $\Gn$ for $n \in \N \cup \{\infty\}$. If  $A_n\underset{\text{Weid}}{\overset{\text{gnrc}}{\to}} A_{\infty}$  then $A_n \underset{\text{QUE}}{\overset{\text{gnrc}}{\to}} A_{\infty}$ with the same convergence speed. If $A_n \underset{\text{QUE}}{\overset{\text{gnrc}}{\to}} A_{\infty}$ with the convergence speed $\delta_n$ then $A_n\underset{\text{Weid}}{\overset{\text{gnrc}}{\to}} A_{\infty}$ with the convergence speed proportional to $\delta_n^{1/2}$. If $A_n \underset{\text{QUE}}{\overset{\text{gnrc}}{\to}} A_{\infty}$ and a parent space exists that factorizes the identification operators, i.e., $J_n$ are of the form $J_n=\iota_\infty^*\iota_n$, in a way that $\iota_n\iota_n^*$ and $\iota_\infty\iota_\infty^*$ commute, then $A_n\underset{\text{Weid}}{\overset{\text{gnrc}}{\to}} A_{\infty}$ with the same convergence speed (modulo multiplicative constant).
\end{theorem}

We now apply the preceding result to our setting. 
Let $(h_n)_{n\in\N}$ be an arbitrary positive sequence such that $\lim_{n\to\infty}h_n=0$. 
We set the sequence of the Hilbert spaces and the corresponding operators to be
\begin{align*}
    \Gn = \mathscr{H}_{h_n}^{2d}, \quad A_n = \tilde{H}_{h_n}\,\text{ or }\, A_n = H^+_{h_n} \quad (\forall n\in \N \cup \{\infty\}) .
\end{align*}

\begin{corollary} \label{cor:different_convergence}
  Let $H_0$, $\tilde H_{h_n}$, and $H^+_{h_n}$ be given by \eqref{eq:Hh0_def}, \eqref{eq:Hh_def}, and \eqref{eq:Hh+_def}, respectively.
    Under the assumptions of Theorem \ref{theo:main} and assuming that the generating functions $\phi_0,\, \psi_0$ used in \eqref{eq:generating_functions} are equal, we have, as $n \to \infty$ ,
    \begin{align*}
        \tilde H_{h_n}\underset{\text{Weid}}{\overset{\text{gnrc}}{\to}} H_0, \quad \text{ and } \quad \tilde H_{h_n}\underset{\text{QUE}}{\overset{\text{gnrc}}{\to}} H_0 \quad \text{with} \quad \delta_{n} = C{h}_n^{\min\{2, 2\gamma-\max\{\beta,0\}-1\}},
        \\
         H^+_{h_n}\underset{\text{Weid}}{\overset{\text{gnrc}}{\to}} H_0, \quad \text{ and } \quad H^+_{h_n}\underset{\text{QUE}}{\overset{\text{gnrc}}{\to}} H_0 \quad \text{with} \quad \delta_{n} = C{h}_n^{\min\{2, 2\gamma-\max\{\beta,0\}-1\}}.
    \end{align*}
\end{corollary}

\begin{proof}

Recall that if $\phi_0=\psi_0$, then $J_{h_n}$ is an isometry and $K_{h_n}=J_{h_n}^*$. Hence, Theorem  \ref{theo:main} immediately yields that $\tilde H_{h_n}\underset{\text{Weid}}{\overset{\text{gnrc}}{\to}} H_0$ and $H^+_{h_n}\underset{\text{Weid}}{\overset{\text{gnrc}}{\to}} H_0$ with $z=\pm\iu$, $\mathscr{G} = \mathscr{G}_{\infty}= \mathscr{H}^{2d}$, $\mathscr{G}_n=\HH_{h_n}^{2d}$, $\iota_\infty=I_{\HH^{2d}}$, $\iota_n=J_{h_n}$, and $\delta_{n}=C{h}_n^{\min\{2, 2\gamma-\max\{\beta,0\}-1\}}$.

\end{proof}

\subsection{Convergence in the sense of Barker}
In his series of papers \cite{Ba_01}, \cite{Ba2_01}, and \cite{Ba3_01}, Barker introduced an abstract framework how to realize dynamical systems on  "continuous" Hilbert spaces as limits of dynamical systems on "discrete" Hilbert spaces. We will demonstrate that the discretization scheme adopted in the present paper fits within this framework. Furthermore, we will investigate the question of how the generalized strong/norm resolvent convergence is related to the convergence of operators in the vein of Barker. We start by recalling a couple of definitions. Since Barker formulated his results for nets, within this subsection, $\NN$ stands for a directed set, $\Gn,\, n\in\NN,$ and  $\GG$ are Hilbert spaces with inner products $\langle\cdot,\cdot\rangle_n$ and $\langle\cdot,\cdot\rangle_\infty$, respectively.

\begin{definition} \label{eq:def_resolution}
Let $\res_n:\Sub\to\Gn,\, n\in\NN,$ be linear mappings on a dense subspace $\Sub$ of $\GG$ such that 
\begin{equation*}
\langle\phi,\chi\rangle_\infty=\lim_{n\in\NN}\langle\res_n\phi,\res_n\chi\rangle_n \quad (\forall{\phi,\chi\in\Sub}).
\end{equation*}
Then $(\Gn,\res_n)_{n\in\NN}$ is called an \emph{inductive resolution} of $\GG$.
\end{definition}

\begin{definition} 
We say that the net $(\chi_n)$ \emph{converges} to $\chi_\infty\in\GG$ in the sense of Barker if and only if the norms $\|\chi_n\|$ are eventually bounded, i.e. bounded for sufficiently large $n$, and for all $\phi\in\Sub$
\begin{equation*}
\langle\phi,\chi_\infty\rangle_\infty=\lim_{n\in\NN}\langle\res_n\phi,\chi_n\rangle_n.
\end{equation*}
If this is the case, we call $\chi_\infty$ the \emph{limit} of $(\chi_n)$ and we write $\chi_\infty=\wtlim_{n\in\NN}\chi_n$.
\end{definition}

Note that given any $\chi\in\Sub$, $\chi=\wtlim_{n\in\NN}\res_n\chi$, see \cite[Remark 2.1]{Ba_01}.

\begin{definition} \label{def:op_conv}
Let $R_n,\, n\in\NN,$ and $R_\infty$ be bounded operators on $\Gn$ and $\GG$, respectively. We say that the net $(R_n)$ converges to $R_\infty$ \emph{in the sense of Barker} if and only if the norms $\|R_n\|$ are eventually bounded and, for all $\chi_\infty\in\GG$ and all nets $(\chi_n)$ such that $\wtlim_{n\in\NN}\chi_n=\chi_\infty$, $R_\infty\chi_\infty=\wtlim_{n\in\NN}(R_n\chi_n)$.
We write $\wtlim_{n\in\NN}R_n=R_\infty$ for brevity.
\end{definition}

Now, let $\NN$ be $(0,+\infty)$ arranged in descending order. Hence, we may write $\wtlim_{h\to 0}$ (the right-hand limit) for $\wtlim_{n\in\NN}$. Furthermore, we will always assume in this subsection that $\phi_0=\psi_0$, which, in particular, implies $K_h=J_h^*$, see the last paragraph of Section \ref{subsec:JK_operators}. Under these conditions we get the following observations. 

\begin{proposition} \label{prop:resolution}
If $\phi_0=\psi_0$ then $(\Hh^{l},K_h)_{h>0}$ is an inductive resolution of $\HH^{l}$. For $\Sub$ from Definition \ref{eq:def_resolution}, we may take $\Sub=\HH^{l}$.
\end{proposition}

\begin{proof}
Using the formula (2.11) from \cite{CoGaJe_21} with $\phi_0=\psi_0$, we get
\begin{multline*}
\langle K_h\phi,K_h\chi\rangle_{\Hh^{l}}=\langle \phi, J_h  K_h\chi\rangle_{\HH^{l}}=\langle\F \phi,\F J_h  K_h\F^*\F\chi\rangle_{\hat\HH^{l}}\\
=(2\pi)^m\int_{\R^m}\hat\psi_0(h\xi)\overline{(\F\phi)(\xi)}\cdot\sum_{j\in\Z^m}\overline{\hat\psi_0(h\xi+2\pi j)}(\F\chi)(\xi+2\pi j/h) \dd\xi\\
=(2\pi)^m\int_{[-\frac{3\pi}{2h},\frac{3\pi}{2h}]^m}\hat\psi_0(h\xi)\overline{(\F\phi)(\xi)}\cdot\sum_{|j|\leq 1}\overline{\hat\psi_0(h\xi+2\pi j)}(\F\chi)(\xi+2\pi j/h) \dd\xi, 
\end{multline*}
for any $\phi,\chi\in\HH^l$. For $\xi\in M_h:=[-\frac{\pi}{2h},\frac{\pi}{2h}]^m$, even the terms with $|j|=1$ vanish and, due to the identity (2.8) from \cite{CoGaJe_21}, $|\hat\psi_0(h\xi)|^2=(2\pi)^{-m}$. Therefore, we may write
\begin{multline*}
(2\pi)^m\int_{M_h}\hat\psi_0(h\xi)\overline{(\F\phi)(\xi)}\cdot\sum_{|j|\leq 1}\overline{\hat\psi_0(h\xi+2\pi j)}(\F\chi)(\xi+2\pi j/h) \dd\xi\\
=\int_{M_h} \overline{(\F\phi)(\xi)}\cdot (\F\chi)(\xi) \dd\xi.
\end{multline*} 
By the dominated convergence theorem, we have
\begin{multline*}
\lim_{h\to 0+}\int_{M_h} \overline{(\F\phi)(\xi)}\cdot (\F\chi)(\xi) \dd\xi=\int_{\R^m} \overline{(\F\phi)(\xi)}\cdot (\F\chi)(\xi) \dd\xi\\
=\langle\F\phi,\F\chi\rangle_{\hat\HH^{l}}=\langle\phi,\chi\rangle_{\HH^{l}}.
\end{multline*}
On the other hand, 
\begin{multline*}
\Big|\int_{[-\frac{3\pi}{2h},\frac{3\pi}{2h}]^m\setminus M_h}\hat\psi_0(h\xi)\overline{(\F\phi)(\xi)}\cdot\sum_{|j|\leq 1}\overline{\hat\psi_0(h\xi+2\pi j)}(\F\chi)(\xi+2\pi j/h) \dd\xi\Big| \\
\leq C \sum_{|j|\leq 1}\int_{\R^m\setminus M_h}|(\F\phi)(\xi)\cdot (\F\chi)(\xi+2\pi j/h)| \dd\xi\\
 \leq C\|\F\phi\|_{L^2(\R^m\setminus M_h;\C^{l})}\|\F\chi\|_{\hat{\HH}^{l}}.
\end{multline*}
Since $\F\phi\in \hat{\HH}^{l}=L^2(\R^m;\C^{l})$, $\lim_{h\to 0+}\|\F\phi\|_{L^2(\R^m\setminus M_h;\C^{l})}=0$. We conclude that
\begin{equation*}
\lim_{h\to 0+}\langle K_h\phi,K_h\chi\rangle_{\Hh^{l}}=\langle\phi,\chi\rangle_{\HH^{l}}.
\end{equation*}
\end{proof}

\begin{corollary} \label{cor:slim}
Let $\phi_0=\psi_0$. Then $\slim_{h\to 0+}J_h K_h=I_{\HH^l}$.
\end{corollary}
\begin{proof}

Since $J_h K_h$ is an orthogonal projection, when $\phi_0=\psi_0$, its weak convergence to $I_{\HH^l}$, which is equivalent to Proposition \ref{prop:resolution}, immediately implies also the strong convergence to $I_{\HH^l}$.
\end{proof}

\begin{proposition} \label{prop:conv_relation}
Let $\phi_0=\psi_0$ and $R_h,\,h>0,$ and $R_0$ be bounded operators on $\Hh^{l}$ and $\HH^{l}$, respectively. If $R_h$ is normal for $h\geq 0$, the norms $\|R_h\|$ are eventually bounded, i.e., bounded for all sufficiently small $h>0$, and $R_h$ converges to $R_0$ in the generalized strong sense, i.e.,
\begin{equation} \label{eq:strong_conv_Rh}
\lim_{h\to 0+} \|(J_h R_h K_h-R_0)\chi\|_{\HH^{l}}=0 \quad (\forall \chi\in\HH^{l}),
\end{equation}
then $\wtlim_{h\to 0}R_h=R_0$ in the sense of Definition \ref{def:op_conv} with the inductive resolution described in Proposition \ref{prop:resolution}. If the latter is true, then $R_{h}$ converges to $R_0$ in the generalized weak sense, i.e.
\begin{equation} \label{eq:gen_weak}
\lim_{h\to 0+}\langle\phi, (J_{h} R_{h} K_{h}-R_0)\chi\rangle_{\HH^{l}}=0 \quad (\forall\phi,\chi\in\HH^{l})
\end{equation}
and $R_{h}^*$ converges to $R_0^*$ in the generalized strong sense, i.e.
\begin{equation} \label{eq:gen_strong_adj}
\lim_{h\to 0+}\|(J_{h} R_{h}^* K_{h}-R_0^*)\chi\|_{\HH^{l}}=0 \quad (\forall\chi\in\HH^{l}).
\end{equation}
If, in addition, $(R_{h})_{h\geq 0}$ are normal, then  $R_{h}$ converges to $R_0$ in the generalized strong sense, too.
\end{proposition}

\begin{proof}
Take any net $(\chi_h)_{h>0}$ such that $\wtlim_{h\to 0}\chi_h=\chi$ and any $\phi\in\HH^{l}$.  Using $K_hJ_h=I_{\Hh^{l}}$, we obtain
\begin{multline} \label{eq:R_hR_0_bound}
|\langle K_h\phi,R_h \chi_h\rangle_{{\Hh^l}}-\langle\phi,R_0\chi\rangle_{\HH^{l}}|=|\langle \phi,J_h R_h K_h J_h\chi_h-R_0\chi\rangle_{\HH^{l}}|\\
\leq |\langle \phi,(J_h R_h K_h-R_0) J_h\chi_h\rangle_{\HH^{l}}|+|\langle\phi, R_0 J_h \chi_h-R_0\chi\rangle_{\HH^{l}}|.
\end{multline}
The first term of this bound can be further estimated as follows,
\begin{multline} \label{eq:1st_term_bound}
|\langle \phi,(J_h R_h K_h-R_0) J_h\chi_h\rangle_{\HH^{l}}|=|\langle (J_h R_h^* K_h-R_0^*)\phi, J_h\chi_h\rangle_{\HH^{l}}|\\
\leq \|(J_h R_h^* K_h-R_0^*)\phi\|_{\HH^{l}}\|J_h\|\|\chi_h\|_{\Hh^{l}}\leq C\|(J_h R_h^* K_h-R_0^*)\phi\|_{\HH^{l}}.
\end{multline}
Here, we used the uniform boundedness of $\|J_h\|$ and $\|\chi_h\|_{\Hh^{l}}$ for $h>0$. Since for a normal $R_h$, $J_h R_h^* K_h$ is normal, too, and a sequence of normal operators converges strongly to a normal operator if and only if the same is true for the corresponding adjoints,
\begin{equation} \label{eq:adjoint_conv}
\lim_{h\to 0+} \|(J_h R_h^* K_h-R_0^*)\phi\|_{\HH^{l}}=0.
\end{equation}
The second term on the right-hand side of \eqref{eq:R_hR_0_bound} obeys
\begin{equation*}
\lim_{h\to 0+}|\langle\phi, R_0 J_h \chi_h-R_0\chi\rangle_{\HH^{l}}|=\lim_{h\to 0+}|\langle K_h R_0^*\phi, \chi_h\rangle_{\Hh^{l}}-\langle R_0^*\phi,\chi\rangle_{\HH^{l}}|=0,
\end{equation*}
because $\wtlim_{h\to 0}\chi_h=\chi$. Putting this, \eqref{eq:R_hR_0_bound}, \eqref{eq:1st_term_bound}, and \eqref{eq:adjoint_conv} together, we arrive at $\wtlim_{h\to 0}R_h\chi_h=R_0\chi$.

To show \eqref{eq:gen_weak}, we just write
\begin{equation*}
\langle\phi, J_h R_h K_h\chi\rangle_{\HH^{l}}=\langle K_h\phi, R_h K_h\chi\rangle_{\Hh^{l}}
\end{equation*}
and use the fact that
\begin{equation*}
\lim_{h\to 0+} \langle K_h\phi, R_h K_h\chi\rangle_{\Hh^{l}}=\langle\phi, R_0\chi\rangle_{\HH^{l}}.
\end{equation*}

Now, we will prove that convergence in the Barker sense implies generalized strong convergence of the adjoints. By definition, $\wtlim_{h\to 0}R_{h}=R_0$ means that for every net $(\chi_{h})_{h>0}$ such that $\chi_{h}\in\HH_{h}^l$ and $\wlim_{h\to 0+}J_{h}\chi_{h}=\chi$ in $\HH^l$ it holds
\begin{equation} \label{eq:wlim}
\wlim_{h\to 0+}J_{h}R_{h}K_{h}(J_{h}\chi_{h})=R_0\chi.
\end{equation}
First, we will show that 
\begin{equation} \label{eq:strong_neg}
\lim_{h\to 0+}\|(R_{h}^*K_{h}-K_{h}R_0^*)\chi\|_{\HH_{h}^l}=0 \quad (\forall \chi\in\HH^l).
\end{equation}
For contradiction, assume that \eqref{eq:strong_neg} is not satisfied, i.e., there exist $\psi\in\HH^l,\, \varepsilon>0,$ and a positive sequence $(\hn)_{n\in\N}$ converging to zero such that for all $h_n$ in this sequence we have
\begin{equation} \label{eq:negation}
\|J_{\hn}(R_{\hn}^*K_{\hn}-K_{\hn}R_0^*)\psi\|_{\HH^l}=\|(R_{\hn}^*K_{\hn}-K_{\hn}R_0^*)\psi\|_{\HH_{\hn}^l}\geq\varepsilon.
\end{equation}
Here, in the first equality we used the fact that $J_{\hn}$ is an isometry. For the same reason, the sequence $(\tilde\chi_{\hn})$ defined by
\begin{equation*}
\tilde\chi_{\hn}:=J_{\hn}\chi_{\hn} \quad \text{with} \quad \chi_{\hn}:=\frac{(R_{\hn}^*K_{\hn}-K_{\hn}R_0^*)\psi}{\|(R_{\hn}^*K_{\hn}-K_{\hn}R_0^*)\psi\|_{\HH_{\hn}^l}} 
\end{equation*}
is bounded in $\HH^l$. Therefore, we can find a subsequence of $(h_n)$, denoted by the same symbol, such that $\wlim_{n\to\infty}\tilde\chi_{\hn}=\wlim_{n\to\infty}J_{\hn}\chi_{\hn}=\chi\in\HH^l$. By our assumption, we get \eqref{eq:wlim}. In particular, we have
\begin{equation*}
\lim_{n\to\infty}\langle J_{\hn}R_{\hn}K_{\hn}\tilde\chi_{\hn}, \psi\rangle_{\HH^l}=\langle R_0\chi,\psi\rangle_{\HH^l},
\end{equation*}
which is equivalent to
\begin{equation} \label{eq:wlim1}
\lim_{n\to\infty}\langle \tilde\chi_{\hn}, J_{\hn}R_{\hn}^*K_{\hn}\psi\rangle_{\HH^l}=\langle \chi,R_0^*\psi\rangle_{\HH^l}.
\end{equation}
Moreover, $\wlim_{n\to\infty}\tilde\chi_{\hn}=\chi$ implies that
\begin{equation} \label{eq:wlim2}
\lim_{n\to\infty}\langle\tilde\chi_{\hn},R_0^*\psi\rangle_{\HH^l}=\langle\chi,R_0^*\psi\rangle_{\HH^l}.
\end{equation}
Subtracting \eqref{eq:wlim1} from \eqref{eq:wlim2} we arrive at
\begin{equation*}
\lim_{n\to\infty}\langle \tilde\chi_{\hn}, (J_{\hn}R_{\hn}^*K_{\hn}-R_0^*)\psi\rangle_{\HH^l}=0.
\end{equation*}
Substituting for $\tilde\chi_{\hn}$, we get
\begin{multline*}
\langle \tilde\chi_{\hn}, (J_{\hn}R_{\hn}^*K_{\hn}-R_0^*)\psi\rangle_{\HH^l}=\langle \chi_{\hn}, (K_{\hn}J_{\hn}R_{\hn}^*K_{\hn}-K_{\hn}R_0^*)\psi\rangle_{\HH_{\hn}^l}\\
=\|(R_{\hn}^*K_{\hn}-K_{\hn}R_0^*)\psi\|_{\HH_{\hn}^l}.
\end{multline*}
Consequently, $\lim_{n\to \infty}\|(R_{\hn}^*K_{\hn}-K_{\hn}R_0^*)\psi\|_{\HH_{\hn}^l}=0$ which contradicts \eqref{eq:negation}.

Since $J_{h}$ is an isometry, \eqref{eq:strong_neg} yields
\begin{equation*} 
\lim_{h\to 0+}\|(J_{h}R_{h}^*K_{h}-J_{h}K_{h}R_0^*)\chi\|_{\HH^l}=0 \quad (\forall \chi\in\HH^l),
\end{equation*}
which together with Corollary \ref{cor:slim} implies \eqref{eq:gen_strong_adj}. In the case when $(R_{h})_{h\geq 0}$  are normal, the generalized strong convergence of $R_{h}$ follows by the same argument that has been used to conclude \eqref{eq:adjoint_conv}.
\end{proof}

\begin{remark} If instead of a net $(R_h)_{h>0}$ we deal with a sequence of bounded operators $(R_{h_n})_{n\in\N}$, where $h_n>0$ and $\lim_{n\to\infty}h_n=0$, the assumption on the eventual boundedness of $\|R_{h_n}\|$ in the first statement of Proposition \ref{prop:conv_relation} may be dropped, as the uniform boundedness principle yields that \eqref{eq:strong_conv_Rh} implies that the norms $\|J_{h_n} R_{h_n} K_{h_n}\|$ are bounded uniformly in $n$. Since  the same holds true for the norms  of $J_{h_n}$ and $K_{h_n}$, we get from
\begin{equation*}
\|R_{h_n}\|=\|K_{h_n} J_{h_n} R_{h_n} K_{h_n} J_{h_n}\|\leq \|K_{h_n}\| \|J_{h_n} R_{h_n} K_{h_n}\| \|J_{h_n}\|
\end{equation*}
that the norms $\|R_{h_n}\|$ are also uniformly bounded.
\end{remark}

\begin{remark}
In the proof of Proposition \ref{prop:conv_relation}, the particular form of $K_h$ and $J_h$ or the choice of the particular Hilbert spaces play no role; we only used the facts that 
\begin{equation} \label{eq:abstract_ass}
J_h^*=K_h,\, (J_h)_{h>0} \text{ are isometries},\,  K_hJ_h=I,\text{ and } \slim_{h\to 0+} J_hK_h=I.
\end{equation}
Therefore, even under these abstract assumptions, the convergence in the Barker sense (with $\res_h=K_h$) of normal operators to a normal limit  is equivalent to generalized strong convergence. 

Let us look at two examples. For both of them, we put $\GG=\Gn$ and $\res_n=K_n=I_{\GG}$ for all $n\in\N$. Then \eqref{eq:abstract_ass}, with the obvious modifications for sequences, is trivially satisfied and the sequence of bounded operators $R_n$ obeys $\wtlim_{n\to\infty}R_n=R_\infty$ if and only if $\|R_n\|$ is uniformly bounded and for every sequence $(\chi_n)$ such that $\wlim_{n\to\infty}\chi_n=\chi$ it holds that $\wlim_{n\to\infty} R_n \chi_n=R_\infty\chi$. Moreover, the generalized weak and strong operator convergence are equivalent to the usual weak and strong operator convergence, respectively. 

%First, take $\GG=\ell^2(\N)$ and 
%\begin{equation*}
%(R_n\psi)^j=\begin{cases}
%\psi^n & j=1, \\
%\psi^1 & j=n, \\
%0 & \text{otherwise}.
%\end{cases}
%\end{equation*}
%Here, $\psi^j$ stands for the $j$-th component of $\psi$. It is straightforward to check that $R_n$ is bounded and self-adjoint, $\wlim_{n\to\infty} R_n=0$ but $\wtlim_{n\to\infty} R_n$ does not exist. To see the latter, it is sufficient to choose $(\chi_n)^j=\delta_n^j$. Then $\wlim_{n\to\infty}\chi_n=0$ but $\lim_{n\to\infty}\langle\psi,R_n \chi_n\rangle=\overline{\psi^1}$, which is generally non-zero. This shows that even in the self-adjoint setting, weak convergence generally does not  imply convergence in the sense of Barker. 

First, take $\GG=\ell^2(\N)$ and
\begin{equation*}
%\label{eq:example_2}
(R_n\psi)^j=\begin{cases}
0 & j\leq n,\\
\psi^{j-n} & j>n,
\end{cases}
\end{equation*}
where $\psi^j$ stands for the $j$-th component of $\psi$.
Note that $R_n$ is not normal, $\|R_n\psi\|=\|\psi\|$, and $(R_n^*\psi)^j=\psi^{j+n}$. Using 
\begin{equation*}
|\langle\psi,R_n\chi_n\rangle|=|\langle R_n^*\psi,\chi_n\rangle|\leq\|R_n^*\psi\|\|\chi_n\|
\end{equation*}
together with 
\begin{equation} \label{eq:ex2}
\slim_{n\to\infty} R_n^*=0
\end{equation}
and the fact that $\|\chi_n\|$ is bounded for any weakly convergent sequence $(\chi_n)$, we infer that $\wtlim_{n\to \infty}R_n=0$. On the other hand, $\slim_{n\to\infty}R_n\neq 0$. This example demonstrates that for non-normal operators, convergence in the Barker sense generally does not imply strong convergence. 

Second, we consider the sequence $(R^*_n)_{n\in\N}$, for which we already have \eqref{eq:ex2}. For $(\chi_n)^j=\delta_n^{j-1}$ we get $\wlim_{n\to\infty}\chi_n=0$, but $\wtlim_{n\to\infty}R^*_n$ does not exist, because
$\langle\psi,R_n^*\chi_n\rangle=\langle R_n\psi,\chi_n\rangle=\overline{\psi^1}$ which is generally non-zero. Therefore, for non-normal operators, strong convergence does not necessarily imply convergence in the Barker sense.
\end{remark}

\begin{corollary} 
Whenever a net of self-adjoint operators converges in the generalized norm resolvent sense, like in Theorem \ref{theo:main}, the net of resolvents themselves converges in the sense of Barker, cf. Definition \ref{def:op_conv}.
\end{corollary}

\subsection{Convergence of operator functions and spectra} 

Once we know that a sequence of operators converges in the QUE-generalized norm resolvent sense, we can employ a variant of functional calculus \cite[Theorem 4.2.9]{Olaf}. In particular, in view of Corollary \ref{cor:different_convergence} we get the following result.

\begin{proposition} \label{prop:functional_calculus}
Let $(h_n)_{n\in\N}\subset[0,h_0)$ be a sequence such that $\lim_{n\to\infty}h_n=0$ and $f\in C([0,+\infty])$. Furthermore, let $H_0$, $\tilde H_{h_n}$, and $H^+_{h_n}$ be given by \eqref{eq:Hh0_def}, \eqref{eq:Hh_def}, and \eqref{eq:Hh+_def}, respectively.
    Under the assumptions of Theorem \ref{theo:main} and assuming that $H_0$, $\tilde H_{h_n}$, and $H^+_{h_n}$ are non-negative and that the generating functions $\phi_0,\, \psi_0$ used in \eqref{eq:generating_functions} are equal, we have, as $n\to\infty$, 
 \begin{equation*}
   f(\tilde H_{h_n}) \underset{\text{QUE}}{\overset{\text{gnrc}}{\to}} f(H_0)\quad \text{and} \quad f( H^+_{h_n}) \underset{\text{QUE}}{\overset{\text{gnrc}}{\to}} f(H_0)
 \end{equation*}
with convergence speed proportional to ${h}_n^{\min\{2, 2\gamma-\max\{\beta,0\}-1\}}$.
\end{proposition}

For a similar result  with $f$ bounded and measurable see \cite[Theorem 4.2.10]{Olaf}. For a related result for operators that are not non-negative or even self-adjoint and a holomorphic $f$ see \cite[Theorem 4.5.10]{Olaf}. Beware that in the case of non-self-adjoint operators, Definition \ref{def:QUE_convergence} has to be modified, cf. \cite[Definition 4.5.6]{Olaf}.

Convergence of spectra is another consequence of QUE-generalized norm resolvent convergence. A detailed discussion on this topic can be found in Sections 4.3 and 4.6 of \cite{Olaf}. In particular, from \cite[Theorems 4.3.3, 4.3.4, and 4.3.5]{Olaf} we immediately get the following result for non-negative self-adjoint operators.

\begin{proposition}
Let $(h_n)_{n\in\N}\subset[0,h_0)$ be a sequence such that $\lim_{n\to\infty}h_n=0$. Furthermore, let $H_0$, $\tilde H_{h_n}$, and $H^+_{h_n}$ satisfy all the assumptions of Proposition \ref{prop:functional_calculus} and let the generating functions $\phi_0,\, \psi_0$ used in \eqref{eq:generating_functions} be equal. Then 
\begin{enumerate}[label=(\alph*)]
\item $\lim_{n\to\infty}\overline{d}(\sigma(\tilde H_{h_n}),\sigma(H_0))=0$ and $\lim_{n\to\infty}\overline{d}(\sigma(H^+_{h_n}),\sigma(H_0))=0$,
\item  $\lim_{n\to\infty}\overline{d}(\sigma_{\textup{ess}}(\tilde H_{h_n}),\sigma_{\textup{ess}}(H_0))=0$ and $\lim_{n\to\infty}\overline{d}(\sigma_{\textup{ess}}(H^+_{h_n}),\sigma_{\textup{ess}}(H_0))=0$,
\item discrete spectrum is \emph{lower semi-continuous}, i.e.,\\
 $\lim_{n\to\infty}\overline{d}_{-}(\sigma_{\textup{d}}(\tilde H_{h_n}),\sigma_{\textup{d}}(H_0))=0$ and $\lim_{n\to\infty}\overline{d}_{-}(\sigma_{\textup{d}}(H^+_{h_n}),\sigma_{\textup{d}}(H_0))=0$.
\end{enumerate}
\end{proposition}

\noindent Here, $\overline{d}$ stands for the "weighted" Hausdorff distance, which is defined as a usual Hausdorff distance but with respect to the metric $\rho:\, \rho(a,b):=|(1+a)^{-1}-(1+b)^{-1}|$ on $[0,+\infty]$. Similarly, $\overline{d}_{-}(A,B)$ stands for the maximal inside distance of $A$ to $B$ derived from the same metric. In particular, (c) means that every discrete eigenvalue of $H_0$ is approximated by discrete eigenvalues of $\tilde H_{h_n}$ and $H^+_{h_n}$, respectively. In fact, the total multiplicity of these discrete eigenvalue is preserved.

\begin{remark}
In our setting, $\sigma_{\textup{d}}(H_0)=\emptyset$. However, one can perturb $H_0$ by a multiplicative potential, that is discretized simply by sampling at the lattice points, to produce discrete eigenvalues. Adapting \cite[Proposition 4.3]{CoGaJe_21} and \cite[Section 7]{CoGaJe_23}, it is then possible to extend Theorem \ref{theo:main} to a certain class of perturbed operators.
\end{remark}

When dealing with operators that are not non-negative (or not even self-adjoint), we may use \cite[Theorem 4.6.4]{Olaf} which in our self-adjoint setting yields the following result.

\begin{proposition}
Let $(h_n)_{n\in\N}\subset[0,h_0)$ be a sequence such that $\lim_{n\to\infty}h_n=0$. Furthermore, let $H_0$, $\tilde H_{h_n}$, and $H^+_{h_n}$ be given by \eqref{eq:Hh0_def}, \eqref{eq:Hh_def}, and \eqref{eq:Hh+_def}, respectively. Finally, let
the assumptions of Theorem \ref{theo:main} be satisfied and the generating functions $\phi_0,\, \psi_0$ used in \eqref{eq:generating_functions} be equal. If $\lambda\in\sigma_{\textup{d}}(H_0)$ is a discrete eigenvalue of multiplicity $\mu$ and, for all $h$ sufficiently small, $\lambda\notin\sigma_{\textup{ess}}(\tilde H_{h_n})$ or $\lambda\notin\sigma_{\textup{ess}}(H^+_{h_n})$, then there exist $\mu$ discrete eigenvalues (counting multiplicities) $\lambda_{j}(h_n),\, j=1,\ldots,\mu,$ of $\tilde H_{h_n}$ or $H^+_{h_n}$, respectively, such that
$$|\lambda_{j}({h_n})-\lambda|=\mathcal{O}\big(h_n^{\min\{2, 2\gamma-\max\{\beta,0\}-1\}})\quad \text{ as } n\to\infty.$$
\end{proposition}

\section*{Acknowledgement}
The authors thank \v{S}imon Kos for turning their attention to the papers of Laurence Barker. They also thank Olaf Post for fruitful discussions. M.T. and R.K. were supported by Grant Nos. 26-21940S and 23-07947S of the Czech Science Foundation, respectively.

\end{document}